\documentclass[letterpaper, 10 pt, journal, onecolumn]{IEEEtran}

\usepackage{graphicx} 
\usepackage{subcaption} 

\usepackage{algorithm}
\usepackage{algorithmic}
\usepackage[top=1in,bottom=1in,left=1in,right=1in]{geometry}

\usepackage{amsthm,graphicx}
\usepackage{amsmath,amssymb}
\usepackage{mathtools}
\usepackage{enumitem}
\usepackage{xcolor}
\usepackage{bm}
\usepackage{soul}
\usepackage{bbold}
\usepackage{dsfont}
\usepackage[mathscr]{euscript}
\usepackage{mathtools,cancel}
\usepackage{stmaryrd} 
\usepackage{cite}
\usepackage{hyperref}

\newcommand{\MC}[1]{\mathcal{#1}} 
\newcommand{\TT}[1]{\texttt{#1}}
\newcommand{\comment}[1]{}

\newcounter{prblm_count}
\newcounter{thm_count}
\newcounter{claim_count}
\newcounter{assump_count}
\newcounter{defn_count}
\newcounter{rmrk_count}

\newcounter{lem_count}

\newtheorem{problem}[prblm_count]{\bf Problem}
\newtheorem{theorem}[thm_count]{\bf Theorem} 
\newtheorem{lemma}[lem_count]{\bf Lemma} 
\newtheorem{corollary}[thm_count]{\bf Corollary} 
\newtheorem{claim}[claim_count]{\bf Claim} 
\newtheorem{assumption}[assump_count]{\bf Assumption}
\newtheorem{definition}[defn_count]{\bf Definition}
\newtheorem{remark}[rmrk_count]{\bf Remark}

\IEEEoverridecommandlockouts                              

\title{\LARGE \bf
Private Learning on Networks: Part II \footnote{\\\textbf{Current Version}: 5 November 2017. \textbf{First Version}: 27 March 2017. \\ \textbf{Comments}: Privacy-Convergence Trade-off added. New simulation results added. \\ \textbf{Acknowledgment}: This research is supported in part by National Science Foundation awards 1421918 and 1610543, and Toyota InfoTechnology Center. Any opinions, findings, and conclusions or recommendations expressed here are those of the authors and do not necessarily reflect the views of the funding agencies or the U.S. government.}
}

\author{Shripad Gade$^{1}$ and Nitin H. Vaidya$^{2}$
\thanks{$^{1}$ PhD candidate in Electrical and Computer Engineering, University of Illinois Urbana-Champaign.
        {\tt\small gade3@illinois.edu}}%
\thanks{$^{2}$ Professor in Electrical and Computer Engineering and Coordinated Science Laboratory,
		University of Illinois Urbana-Champaign, 1308 W Main St., Urbana, 61801.
        {\tt\small nhv@illinois.edu}}%
}

\begin{document} 
\maketitle

\begin{abstract} 
This paper considers a distributed multi-agent optimization problem, with the global objective consisting of the sum of local objective functions of the agents. The
agents solve the optimization problem using local computation and communication between adjacent agents in the network. We present two randomized iterative algorithms for distributed optimization. To improve privacy, our algorithms add ``structured'' randomization to the information exchanged between the agents. We prove deterministic correctness (in every execution) of the proposed algorithms despite the information being perturbed by noise with non-zero mean. We prove that a special case of a proposed algorithm (called function sharing) preserves privacy of individual polynomial objective functions under a suitable connectivity condition on the network topology.
\end{abstract} 

\section{Introduction}
Distributed optimization has received a lot of attention in the past couple of decades. It involves a system of networked agents that optimize a global objective function $f(x) \triangleq \sum_i f_i(x)$, where $f_i(x)$ is the local objective function of agent $i$. Each agent is initially only aware of its own local objective function. The agents solve the global optimization problem in an iterative manner. Each agent maintains ``a state estimate'', which it shares with its neighbors
in each iteration, and then updates its state estimate using the information received from the neighbors. A distributed optimization algorithm must ensure 
that the state estimates maintained by the agents converge to an optimum of the global cost function.

Emergence of networked systems has led to the application of distributed optimization framework in several interesting contexts, such as machine learning, resource allocation and scheduling, and robotics \cite{xiao2006optimal,rabbat2004distributed}.  
In a distributed machine learning scenario, partitions of the dataset are stored among several different agents (such as servers or mobile devices \cite{konevcny2015federated}), and these agents solve a distributed optimization problem in order to collaboratively learn the most appropriate ``model parameters''.
In this case, $f_i(x)$ at agent $i$ may be a {\em loss function} computed over
the dataset stored at agent $i$, for a given choice $x$ of the {\em model parameters} (i.e., here $x$ denotes a vector of model parameters).

Distributed optimization can reduce communication requirements of learning, since the agents communicate information that is often much smaller in size than each agent's local dataset that characterizes 
its local objective function. 
The scalability of distributed optimization algorithms, and their applicability
for geo-distributed datasets, have made them a desirable choice for distributed
learning \cite{kraska2013mlbase,NIPS2014_5597,cano2016towards}. 

Distributed optimization algorithms rely on exchange of information between agents, making them vulnerable to privacy violations. In particular, in case of distributed learning, the local objective function
of each agent is derived using a local dataset known only to that agent. Through the information
exchanged between agents, information about an agent's local dataset may become known to other agents.
Therefore, privacy concerns have emerged as a critical challenge in distributed optimization \cite{shokri2015privacy,pasqualetti2012cyber}.  

In this paper we present two algorithms that use ``structured randomization'' of state
estimates shared between agents. In particular, our structured randomization approach 
obfuscates the state estimates by adding {\em correlated random noise}. Introduction of random noise into the state estimates allows the agents to improve privacy. {\em Correlation} (as elaborated later) helps
to ensure that our algorithms asymptotically converge to a true optimum, despite perturbation of state
estimates with non-zero mean noise. We also prove strong privacy guarantees for a special case of our algorithm for a distributed polynomial optimization problem. Contributions of this paper are as follows:

\begin{itemize}
\item We present {\em Randomized State Sharing} (\TT{RSS}) algorithms for distributed optimization that use structured randomization. Our first algorithm, named \TT{RSS-NB}, introduces noise
that is {\em Network Balanced} (NB), as elaborated later, and the second algorithm, \TT{RSS-LB},
introduces {\em Locally Balanced} (LB) noise. We prove {\em deterministic} convergence (in every execution) to an optimum, despite the use of randomization.

\item We consider a special case of \TT{RSS-NB} (called ``Function Sharing'' or \TT{FS}), where the random perturbations added to local iterates are state-dependent. State-dependent random perturbations simulate the obfuscation of objective function using a {\em noise function}. We argue that the \TT{FS} algorithm achieves a strong notion of privacy.

\item We use \TT{RSS-NB} and \TT{RSS-LB} algorithms to train a deep neural network for digit recognition using the MNIST dataset, and to train a logistic regression model for document classification of the Reuters dataset. The experiments validate our theoretical results and we show that we can obtain high accuracy models, despite introducing randomization to improve privacy.

\end{itemize}%

\noindent
\textbf{Related Work: }
Many {distributed optimization} algorithms have appeared in the literature over the past decade, including Sub-gradient Descent \cite{nedic2009distributed,ram2010distributed}, Dual Averaging \cite{duchi2012dual}, Incremental Algorithms \cite{rabbat2004distributed,  ram2009incremental}, Accelerated Gradient \cite{jakovetic2014fast}, ADMM \cite{shi2014linear} and EXTRA \cite{shi2015extra}. Solutions to distributed optimization of convex functions have been proposed for myriad scenarios involving directed graphs \cite{Nedi2015, xi2015linear}, communication link failures and losses \cite{hadjicostis2016robust}, asynchronous communication models \cite{nedic2011asynchronous,liu2015asynchronous,wei20131}, and stochastic objective functions \cite{agarwal2011distributed,ram2010distributed}. 

Privacy-preserving methods for optimization and learning can be broadly classified into cryptographic approaches and non-cryptographic approaches \cite{weeraddana2013per}. Cryptography-based privacy preserving optimization algorithms \cite{pinkas2002cryptographic,hong2016privacy} tend to be computationally expensive. Non-cryptographic approaches have gained popularity in recent years. \emph{$\epsilon$-differential privacy} is a probabilistic technique that involves use of randomized perturbations \cite{huang2015differentially,nozari2015differentially,abadi2016deep,Han2016,hamm2016learning} to minimize the probability of uncovering of specific records from databases. Differential privacy methods, however, suffer from a fundamental trade-off between the accuracy of the solution and the privacy margin (parameter $\epsilon$) \cite{Han2016}. \emph{Transformation} is another non-cryptographic technique that involves converting a given optimization problem into a new problem via algebraic transformations such that the solution of the new problem is the same as the solution of the old problem \cite{mangasarian2012privacy,wang2011secure}. This enables agents to conceal private data effectively while the quality of solution is preserved. Transformation approaches in literature, however, cater only to a relatively small set of problem classes. 

\section{Notation and Problem Formulation} \label{Sec:Notation}

We consider a {\em synchronous} system consisting of $n$ agents connected using a network of undirected (i.e., bidirectional) communication links. The communication links are always reliable. The set of agents is denoted by $\mathcal{V}$; thus, $|\mathcal{V}| = n$.

Although all the links are undirected, for convenience, we represent each undirected link using 
a pair of directed edges. Define $\mathcal{E}$ as a set of directed edges corresponding to the communication links in the network:\\ $\mathcal{E} = \{(u,v) : u,v \in \mathcal{V} \text{ and } u \text{ communicates with } v \}$.\\ Thus, the communication network is represented using a graph $\mathcal{G} = (\mathcal{V}, \mathcal{E})$. The neighbor set of agent $v$ is defined as the set of agents that are connected to agent $v$. By convention, $\MC{N}_v$ includes $v$ itself, i.e. $\MC{N}_v = \{u~|(u,v) \in \MC{E}\} \cup \{v\}$. 

We assume that the communication graph $\MC{G}$ is strongly connected. We impose additional connectivity constraint later when analyzing privacy in Section~\ref{Sec:Results}.

The focus of this paper is on iterative algorithms for distributed optimization.
Each agent maintains a {\em state estimate}, which is updated in each iteration of the algorithm. The state estimate at agent $i$ at the start of iteration $k$ is denoted by $x^i_{k}$.
We assume that argument $x$ of $f_i(x)$ is constrained to be in a feasible set $\mathcal{X}\subset \mathbb{R}^D$.
The state estimate of each agent is initialized to an arbitrary vector in $\mathcal{X}$.
For $z\in \mathbb{R}^D$, we define projection operator $\MC{P}_\MC{X}$ as, $$\MC{P_X}(z) ~=~ \arg\min_{y\in\MC{X}}~\|z - y\|.$$
Problem~\ref{Prob:DistLearn} below formally defines the goal of distributed optimization.

\begin{problem} \label{Prob:DistLearn}
Given local objective function $f_i(x)$ at each agent $i\in\mathcal{V}$,
and feasible set $\mathcal{X} \subset \mathbb{R}^D$ (i.e., the set of feasible $x$),
design a distributed algorithm such that, for some
\[ x^* \in {\arg\min}_{x \in \mathcal{X}} \, \sum_{i=1}^n f_i(x),\]
we have
\[ \lim_{k\rightarrow\infty}~ x^i_k = x^*, ~~~~\forall i\in \mathcal{V}.\]
\end{problem}

\noindent
Let $f^*$ denote the optimal value of $f(x)$, i.e. $$f^* = \inf_{x \in \MC{X}} f(x).$$ Let $\MC{X}^*$ denote the set of all optima of $f(x)$, i.e., $$\MC{X}^* = \{x~|~x \in \MC{X}, f(x) = f^*\}.$$ Let $\|.\|$ denote the Euclidean norm. For any matrix $A$, $\|A\|_2 = \sqrt{\lambda_{\max}(A^\dagger A)}$, where $A^\dagger$ denotes conjugate transpose of matrix $A$, and $\lambda_{\max}$ is the maximum eigenvalue.
We make the following assumptions.
\begin{assumption}[Objective Function and Feasible Set] \label{Asmp:FunSet} \leavevmode
\begin{enumerate} 
\item The feasible set, $\MC{X}$, is a non-empty, convex, and compact subset of $\mathbb{R}^D$. \label{Asmp:Set}
\item  The objective function $f_i : \MC{X} \rightarrow \mathbb{R},$ $\forall i\in\mathcal{V}$, is a convex function. Thus, $f(x) := \sum_{i=1}^n f_i(x)$ is also a convex function. \label{Asmp:Function} 
\item The set of optima, $\MC{X}^*$, is non-empty and bounded.
\end{enumerate}
\end{assumption}

\begin{assumption} [Gradient Bound and Lipschitzness] \label{Asmp:GradientCond} \leavevmode
\begin{enumerate} 
\item The gradients are norm-bounded, i.e., $\exists$ $L > 0$ such that, $\| \nabla f_i(x) \| \leq L,$ $\forall \ x \in \MC{X}$ and $\forall \ i\in\mathcal{V}$.\label{Asmp:SubBound}
\item The gradients are Lipschitz continuous, i.e., $\exists$ $N > 0$ such that, $\|\nabla f_i(x) - \nabla f_i(y) \| \leq N \| x -y\|$, $\forall \ x,y \in \MC{X}$ and $\forall \ i\in\mathcal{V}$. \label{Asmp:GradLip}
\end{enumerate}%
\end{assumption}%

\section{Distributed Algorithms} \label{Sec:Algo}
This section first presents an iterative Distributed Gradient Descent algorithm (\TT{DGD}) from prior literature \cite{nedic2009distributed}. Later we modify \TT{DGD} to improve privacy. In particular, we present two algorithms based on {\em Randomized State Sharing} (\TT{RSS}).

\subsection{\TT{DGD} Algorithm \cite{nedic2009distributed}}

Iterative distributed algorithms such as Distributed Gradient Descent (\TT{DGD}) use a combination of consensus dynamics and local gradient descent to distributedly find a minimizer of $f(x)$. More precisely, in each iteration, each agent receives state estimates from its neighbors and performs a consensus step followed by descent along the direction of the gradient of its
local objective function. 

The pseudo-code for the DGD algorithm is presented below as Algorithm~\ref{Algo:DSG}.
The algorithm presents the steps performed by any agent $j\in\MC{V}$.
The different agents perform their steps in parallel. Lines 4-5 are intentionally left blank in Algorithm~\ref{Algo:DSG}, to facilitate comparison with other algorithms presented later in the paper.

As shown on Lines 6 and 7 of Algorithm~\ref{Algo:DSG}, in the $k$-th iteration, each agent $j$ first sends its current
estimate $x_k^j$ to the neighbors, and then receives the estimates from all its
neighbors.  
Using these estimates, as shown on line 8, each agent performs a 
\textit{consensus step} (also called \textit{information fusion}), which
involves computing a convex combination of the state estimates. The resulting convex
combination is named $v_k^j$. Matrix $B_k$ used in this step is a doubly stochastic matrix \cite{nedic2009distributed}, which
can be constructed by the agents using previously proposed techniques,\footnote{$B_k$ has the property that
entries $B_k[i,j]$ and $B_k[j,i]$ are non-zero if and only if $i\in \MC{N}_j$. Recall that the
underlying network is assumed to consist of bidirectional links. Therefore,
$i\in\MC{N}_j$ implies $j\in\MC{N}_i$.} such as Metropolis weights \cite{xiao2006distributed}. The Metropolis weights are:
\begin{align*}
B_k[i,j] = \begin{cases}
1/(1+\max(|\MC{N}_i|,|\MC{N}_j|)) & \text{if } j \in \MC{N}_i\\
1 - \sum_{l \neq i} B_k[i,l] & \text{if } i = j \\
0 & \text{otherwise}
\end{cases}
\end{align*}


Agent $j$ performs \textit{projected gradient descent} step (Line 9, Algorithm~\ref{Algo:DSG}) involving descent from $v_k^j$ along the local objective function's gradient $\nabla f_j(v^j_k)$, followed by projection onto the feasible set $\MC{X}$. This step yields the new state estimate at agent $j$, namely, $x_{k+1}^j$. $\alpha_k$ used on line 9 is called the {\em step size}. The sequence $\alpha_k$, $k\geq 1$, is a non-increasing
sequence such that $\sum_{k=1}^\infty\alpha_k=\infty$ and $\sum_{k=1}^\infty \alpha_k^2<\infty$.


Prior work \cite{nedic2009distributed} has shown that \TT{DGD} Algorithm \ref{Algo:DSG} 
solves Problem 1, that is, the agents' state estimates asymptotically reach consensus on an optimum
in $\MC{X}^*$.

\TT{DGD} is not designed to be privacy-preserving and an adversary may learn
information about an agent's local objective function by observing information exchange
between the agents. We now introduce algorithms that {\em perturb} the state estimates
before the estimates are shared between the agents. The perturbations are intended to hide the true state estimate values and improve privacy. 

\setcounter{algorithm}{-1}
	\begin{algorithm}[t]
	\caption{\TT{DGD} Algorithm \cite{nedic2009distributed}}
	\begin{algorithmic}[1]
	\STATE Input: $\alpha_k~(k\geq 1)$, MAX\_ITER.\\
		Initialization: $x_1^j\in\MC{X},~\forall j\in\MC{V}$. 
	\STATE The steps performed by each agent $j\in\MC{V}$:
	\FOR {$k$ = 1 to MAX\_ITER} 
		\STATE
		\STATE
		\STATE Send estimate $x^j_{k}$ to each agent $i\in \MC{N}_j$
		\STATE Receive estimate $x^i_{k}$ from each agent $i\in\mathcal{N}_j$ 
        \STATE Information Fusion: \\ $ \qquad v^j_{k} = \sum_{i \in \mathcal{N}_j} B_k[j,i] x^i_{k}$  
       	\STATE Projected Gradient Descent: \\ $\qquad x^j_{k+1} = \mathcal{P}_\mathcal{X} \left[ v^j_{k} - \alpha_k \nabla {f}_j(v^j_{k})\right]$
	\ENDFOR
\end{algorithmic}
\label{Algo:DSG}
\end{algorithm}

\subsection{\TT{RSS-NB} Algorithm}
The first proposed algorithm, named {\em Randomized State Sharing$-$Network Balanced} (\TT{RSS-NB})
is a modified version of Algorithm~\ref{Algo:DSG}.
The pseudo-code for algorithm \TT{RSS-NB} is presented as Algorithm~\ref{Algo:PPDOP2} below.

Random variables $s_k^{i,j}$ are used to compute the perturbations.
We will discuss the procedure for computing the perturbation after describing the rest
of the algorithm.
As we will discuss in more detail later, on Line 4
of Algorithm~\ref{Algo:PPDOP2}, agent $j$ computes perturbation
$d_k^j$ to be used in iteration $k$.
On Line 5, the perturbation is weighted by step size $\alpha_k$ and added to state estimate
$x_k^j$ to obtain the perturbed state estimate $w_k^j$ of agent $j$. That is,
\begin{align}
w_k^j &= x_k^j + \alpha_k d^{j}_k.\label{Eq:Perturb}
\end{align}
$\alpha_k$ here is the step size, which is also used in the information fusion step in Line 8.
Properties satisfied by $\alpha_k$ are identical to those in the DGD Algorithm \ref{Algo:DSG}.

Having computed the perturbed estimate $w_k^j$, 
each agent then sends the perturbed estimate $w_k^j$ to its neighbors (Line 6) and receives the perturbed
estimates of the neighbors (Line 7, Algorithm~\ref{Algo:PPDOP2}). Similar to Algorithm \ref{Algo:DSG},
Steps 8 and 9 of the \TT{RSS-NB} algorithm also perform information fusion using
a doubly stochastic matrix $B_k$, followed
by projected gradient descent. 

Now we describe how the perturbation $d_k^j\in \mathbb{R}^D$ is computed on Line 4
of Algorithm~\ref{Algo:PPDOP2}.
The strategy for computing the perturbation is
motivated by a secure distributed averaging algorithm in \cite{abbe2012privacy}.
In iteration $k$, the computation of the perturbation $d_k^j$ at agent $j$
uses variables $s_k^{j,i}$ and $s_k^{i,j}$, $i\in\MC{N}_j$, which take
values in $\mathbb{R}^D$. As shown on Line 4, the perturbation $d_k^j$ is computed
as follows.
\begin{align}
    d^j_k = \sum_{i \in \MC{N}_j} s^{i,j}_k - \sum_{i \in \MC{N}_j} s^{j,i}_k. \label{Eq:SMCPerturbationGen}
\end{align}
Initially, as shown on Line 1, $s_1^{i,j}=s_1^{j,i}$ is the 0 vector (i.e., all elements 0)
for all $i\in\MC{N}_j$. Thus, the perturbation $d_1^j$ computed in iteration 1 is also the 0 vector.
As shown on Line 6 of Algorithm~\ref{Algo:PPDOP2},
in iteration $k\geq 1$, agent $j$ sends to each neighbor $i$ a random vector $s_{k+1}^{j,i}$
and then (on Line 7) it receives random vector $s_{k+1}^{i,j}$ from each neighbor $i$.
These random vectors are then used to compute the random
perturbations in Line 4 of the next iteration.
Due to the manner in which $d_k^j$ is computed, we obtain the following invariant for all iterations $k\geq 1$.
\begin{align}
\sum_{j\in\MC{V}}\,d_k^j~=~0. \label{Eq:RSSNB-NoiseChar}
\end{align}
The distribution from which the random vectors $s_k^{j,i}$ are drawn affects the privacy achieved with this algorithm.
In our analysis, we will assume that
$\|s^{j,i}_k\| \leq \Delta/(2n)$, for all $i,j,k$, where constant $\Delta$ is a parameter of the algorithm, and $n=|\MC{V}|$ is the number of agents.
Procedure for the Computation of perturbation $d_k^j$, as shown in (\ref{Eq:SMCPerturbationGen}), then
implies that $\|d_k^j\|\leq \Delta$.
As elaborated later, there is a trade-off between privacy and convergence rate of the algorithm, with larger
$\Delta$ resulting in slower convergence rate.

%

\begin{figure*}[h]
\begin{minipage}[t]{3.2in}
  \vspace{0pt}      
    \begin{algorithm}[H]
    \caption{\TT{RSS-NB} Algorithm}
    \begin{algorithmic}[1]
    
    \STATE Input: $\alpha_k~(k\geq 1)$, MAX\_ITER.\\
	 Initialization: $x^j_1\in\MC{X},~\forall j\in\MC{V}$ and\\
	\hspace*{0.3in} $s_1^{i,j} =s_1^{j,i}= 0,~\forall j\in\MC{V}, i\in\MC{N}_j$.
    \STATE The steps performed by each agent $j\in\MC{V}$:
    \FOR{$k$ = 1 to MAX\_ITER} 
        \STATE Compute perturbation $d^{j}_k$: \\ $\quad$ $d^j_k = \sum_{i \in \MC{N}_j} s^{i,j}_k - \sum_{i \in \MC{N}_j} s^{j,i}_k$
        \STATE Compute perturbed state $w_k^j$: \\  $\quad w_k^j = x_k^j + \alpha_k d^{j}_k $
        \STATE Send $w_k^j$ and a random vector $s^{j,i}_{k+1}$ to $i\in \MC{N}_j$.
        \STATE Receive $w_k^i$ and $s^{i,j}_{k+1}$ from $i\in \MC{N}_j$.
        \STATE Information Fusion:\\ $\qquad v^j_{k} = \sum_{i \in \mathcal{N}_j} B_k[j,i] w^i_{k}$  
        \STATE Projected Gradient Descent: \\  $\qquad x^j_{k+1} = \mathcal{P}_\mathcal{X} \left[ v^j_{k} - \alpha_k \nabla f_j(v^j_{k})\right]$
    \ENDFOR 
    \end{algorithmic}
    \label{Algo:PPDOP2}
    \end{algorithm}
\end{minipage}%
\hfill
\begin{minipage}[t]{3.2in}
  \vspace{0pt}
  \begin{algorithm}[H]
    \caption{\TT{RSS-LB} Algorithm}
    \begin{algorithmic}[1]
    \STATE Input: $\alpha_k~(k\geq 1)$, MAX\_ITER.\\
		Initialization: $x_0^j\in\MC{X},~\forall j\in\MC{V}$.\\
		$\qquad$ 
    \STATE The steps performed by each agent $j\in\MC{V}$:
    \FOR{$k$ = 1 to MAX\_ITER} 
        \STATE Choose random vector $d^{j,i}_k$, $i\in\MC{N}_j$, such that, \\ $\qquad$ $\sum_{i\in\MC{N}_j}B_k[i,j] d^{j,i}_k = 0. $
        \STATE Compute perturbed state $w_k^{j,i}$: \\  $\qquad w_k^{j,i} = x_k^j + \alpha_k d^{j,i}_k $
        \STATE Send $w_k^{j,i}$ to each $i\in \MC{N}_j$.
        \STATE Receive $w_k^{i,j}$ from each $i\in \MC{N}_j$.
        \STATE Information Fusion: $\qquad \qquad \qquad \qquad $ \\ $\qquad v^j_{k} = \sum_{i \in \mathcal{N}_j} B_k[j,i] w^{i,j}_{k}$  
        \STATE Projected Gradient Descent: \\  $\qquad x^j_{k+1} = \mathcal{P}_\mathcal{X} \left[ v^j_{k} - \alpha_k \nabla f_j(v^j_{k})\right]$
    \ENDFOR 
    \end{algorithmic}
    \label{Algo:PPDOP3}
    \end{algorithm}
\end{minipage}
\end{figure*}

\subsection{\TT{RSS-LB} Algorithm}

Our second algorithm is called {\em Randomized State Sharing$-$Locally Balanced} algorithm (\TT{RSS-LB}). Recall that in \TT{RSS-NB} Algorithm \ref{Algo:PPDOP2}, each agent shares an identical perturbed estimate
with its neighbors. Instead, in \TT{RSS-LB}, each agent shares potentially distinct perturbed state
estimates with different neighbors.
The pseudo-code for \TT{RSS-NB} is presented as Algorithm~\ref{Algo:PPDOP3}.

On Line 4 of Algorithm~\ref{Algo:PPDOP3}, in iteration $k$, agent $j$ chooses a noise vector
$d_k^{j,i}\in\mathbb{R}^D$ for each $i\in\MC{N}_j$ such that
$d^{j,j}_k = 0$, $\|d^{j,i}_k\|\leq \Delta$, where constant $\Delta$ is a parameter of the algorithm,
and
\begin{eqnarray}
\sum_{i \in \MC{N}_j}B_k[i,j]d^{j,i}_k = 0. \label{eq:lb}
\end{eqnarray}
For convenience, for $i\not\in \MC{N}_j$, define $d^{j,i}_k = 0$, that is, the perturbations for non-neighbors are zero. Here, matrix $B_k$ is identical to that used in the information fusion step in Line 8.
Observe that each agent $j$ uses $B_k[j,i]$, $i\in\MC{N}_j$, in the information fusion step,
and $B_k[i,j]$, $i\in\MC{N}_j$, in the computation of above noise vectors. In both cases, the matrix elements
used by agent $j$ correspond only to its neighbors in the network.
Since the random vectors generated by each agent $j$ are {\em locally balanced}, as per (\ref{eq:lb}) above,
the agents do not need to cooperate in generating the perturbations (unlike the \TT{RSS-NB} algorithm).

Using $d_k^{j,i}$ as the perturbation for neighbor $i$,
in Line 5 of Algorithm~\ref{Algo:PPDOP3},
agent $j$ computes the perturbed state estimate $w^{j,i}_k$
to be sent to neighbor $i$, as follows.
\begin{align}
w^{j,i}_k = x^j_k + \alpha_k d^{j,i}_k. \label{Eq:Perturb-RSSLB}
\end{align}
$\alpha_k$ here is the step size, which is also used in the information fusion step in Line 8.
Properties satisfied by $\alpha_k$ are identical to thos in the DGD Algorithm \ref{Algo:DSG}.

Next, in Lines 6 and 7 of Algorithm \ref{Algo:PPDOP3}, agent $j$ sends $w^{j,i}_k$ to each neighbor $i$ and receives perturbed estimate $w^{i,j}_k$ from each neighbor $i$. Agent $j$ performs the information fusion step in Line 8 followed by projected gradient descent in Line 9, similar to the previous algorithms.

\comment{+++++++++++++++++++++++++++++++++++++++++++++++++++++++++++++++++++++
++++++++++++ redraw separate figures for the two algorithms, or just show figure for one algorithm ++++++++++
Figure~\ref{Fig:Schematic} shows perturbed state exchange for \TT{RSS-LB}. Per definition, $w^{j,j}_k = x^j_k$. This implies that perturbed estimates are shared with neighbors and never used by the agent himself (unlike \TT{RSS-NB} algorithm). 
++++++++++++++++ rewrite this paragraph ++++++++++

\begin{figure}[!t]
\centering
\includegraphics[width=0.7\columnwidth]{NIPS_Schematic_New-snip}
\caption{\TT{RSS-NB} and \TT{RSS-LB} perturbations.}
\label{Fig:Schematic}
\end{figure}
+++++++++++++++++++++++++++++++++++++}

\subsection{\TT{FS} Algorithm} \label{Sec:FSAlgo}

The {\em function sharing} algorithm \TT{FS} presented in this section can be viewed as a special case
of the \TT{RSS-NB} algorithm. In this special case of \TT{RSS-NB}, the random vector $s_k^{j,i}$ computed by agent $j$
 is a function of its state estimate $x_k^j$, where the function is independent of $k$.
Thus, the function sharing algorithm uses {\em state-dependent} random vectors.
The pseudo-code for function sharing is presented in Algorithm \ref{Algo:PPDOP1} below
using random functions, instead of state-dependent random vectors. However, the behavior
of Algorithm \ref{Algo:PPDOP1} is equivalent to using state-dependent noise in \TT{RSS-NB}.

In Line 1 of Algorithm \ref{Algo:PPDOP1}, each agent $j$ selects a function $s^{j,i}(x)$ to be sent
to neighbor $i$ in Line 2. These functions are exchanged by the agents. Agent $j$ then uses
them in Line 3 to compute the noise function, which is, in turn, used to compute an obfuscated
local objective function $\widehat{f}_j(x)$.   
Finally, the agents perform DGD Algorithm \ref{Algo:DSG} with each agent $j$
using $\widehat{f}_j(x)$ as its objective function.
We assume that $s^{j,i}(x)$ have bounded and Lipschitz gradients. This implies the obfuscated functions $\widehat{f}_j(x)$ satisfy assumption~A2.
The obfuscated objective function $\widehat{f}_j(x)$ is not necessarily convex. Despite this,
the correctness of this algorithm can be proved using the following observations:
\begin{eqnarray}
\sum_{j\in\MC{V}}~ p_j(x) = 0, \mbox{~~~~~~and}\\
\sum_{j\in\MC{V}} \widehat{f}_j(x) = \sum_{j\in\MC{V}} f_j(x) 
~=~ f(x)
\end{eqnarray}

\begin{algorithm}[t!]
\caption{Function Sharing (\TT{FS}) Algorithm}
\begin{algorithmic}[1]
\STATE The steps performed by each agent $j \in \MC{V}$
\STATE Select a function $s^{j,i}(x)$, $\forall i\in\MC{N}_j$.
\STATE Agent $j$ sends function $s^{j,i}(x)$ to each $i\in\MC{N}_j$.
\STATE Agent $j$ computes a noise function $p_j(x)$ and then the obfuscated local objective
function $\widehat{f}_j(x)$ as follows:
\begin{align}
    p_j(x) &= \sum_{i \in \MC{N}_j} s^{i,j}(x) - \sum_{i \in \MC{N}_j} s^{j,i}(x). \label{Eq:SMCFunctionGen} \\
	\widehat{f}_j(x) &\triangleq f_j(x) + p_j(x) \label{Eq:F_New}
\end{align}%
\vspace{-0.1in}
\STATE Perform \TT{DGD} (Algorithm \ref{Algo:DSG}) wherein agent $j$ uses $\widehat{f}_j(x)$ as its local objective function instead of $f_j(x)$.
\end{algorithmic}
\label{Algo:PPDOP1}
\end{algorithm}%

Effectively, Algorithm~\ref{Algo:PPDOP1} minimizes a {\em convex sum of non-convex functions}. Distributed optimization of a convex sum
of non-convex functions, albeit with an additional assumption of {\em strong convexity} of $f(x)$, was also addressed
in \cite{kvaternik2011lyapunov}, wherein the correctness is shown using Lyapunov stability arguments.
However, \cite{kvaternik2011lyapunov} also does not address how privacy may be achieved.
Additionally, our approach for improving privacy is more general than function sharing, as exemplified by algorithms \TT{RSS-NB} and \TT{RSS-LB}.

\section{Main Results} \label{Sec:Results}

The specification of Problem \ref{Prob:DistLearn} in Section \ref{Sec:Notation} identifies the requirement
for correctness of the proposed algorithms. 
The proof of Theorem \ref{Th:ConvPPDOP2} below is outlined in Section \ref{Sec:Analysis} and presented in detail in Appendix~\ref{Sec:ProofTh1}.

%
\begin{theorem} \label{Th:ConvPPDOP2}
Under Assumptions~\ref{Asmp:FunSet} and \ref{Asmp:GradientCond},
\TT{RSS-NB} Algorithm~\ref{Algo:PPDOP2}, \TT{RSS-LB} Algorithm~\ref{Algo:PPDOP3} and \TT{FS} Algorithm~\ref{Algo:PPDOP1} solve distributed optimization Problem \ref{Prob:DistLearn}. 
\end{theorem}%

Theorem \ref{Th:ConvPPDOP2} implies that the sequence of iterates $\{x^j_{k}\}$, generated by
each agent $j$ converges to an optimum in $\mathcal{X}^*$ asymptotically, despite the introduction of
perturbations.

Now we discuss privacy improvement achieved by our algorithms.
We consider an adversary that compromises a set of up to $f$ agents, denoted as $\MC{A}$ (thus, $|\MC{A}| \leq f$). The adversary can observe everything that each agent in $\MC{A}$ observes.
In particular, the adversary has the knowledge of the local objective functions of agents in $\MC{A}$,
their state, and their communication to and from all their neighbors. Furthermore, the adversary knows the
network topology. 

The goal here is to prevent the adversary from learning the local objective function of any agent $i\not\in\MC{A}$.
%
The introduction of perturbations in the state estimates helps improve privacy, by creating an ambiguity in the following sense.
To be able to exactly determine $f_i(x)$ for any $i\not\in\MC{A}$, the adversary's observations of the
communication to and from agents in $\MC{A}$ has to be compatible with the actual $f_i(x)$, but not with any other 
possible choice for the local objective function of agent $i$. The larger the set of feasible local objective functions of agent $i$
that are compatible with the adversary's observations, greater is the ambiguity.  The introduction of noise naturally increases this ambiguity,
with higher $\Delta$ (noise parameter) resulting in greater privacy. However, this improved privacy comes with a performance cost, as Theorem \ref{Th:FiniteTimeRes} will show. 
%
%
Before we discuss Theorem \ref{Th:FiniteTimeRes}, we first present more precise claims for privacy for the \TT{FS} algorithm.

~

\noindent{\bf Privacy Claims:}
Let $\MC{F}$ denote the set of all feasible instances of Problem \ref{Prob:DistLearn}, characterized by sets of local objective functions. Thus, each element of $\MC{F}$, say $\{g_1(x),g_2(x),\cdots,g_n(x)\}$ corresponds to an instance of
Problem~\ref{Prob:DistLearn}, where the $g_i(x)$ become the local objective functions for each agent $i$. When each agent's local objective function is restricted to be any polynomial
of a bounded degree, the set of feasible functions forms an additive group. Theorem~\ref{Th:TPriv-2} makes a claim regarding the
privacy achieved using {\em function sharing} in this case.

\begin{definition}\label{Def:Compatible}
Recall that $\MC{F}$ is the set of all possible instance of Problem \ref{Prob:DistLearn}.
The adversary's observations are said to be compatible with problem instance $\{g_1(x),g_2(x),\cdots,g_n(x)\}\in\MC{F}$ if the information available to the adversary may be produced when agent $i$'s local objective function is $g_i(x)$
for each $i\in\MC{V}$.
\end{definition}

\begin{theorem}
\label{Th:TPriv-2}
Let the local objective function of each agent be restricted to be a
polynomial of a bounded degree. Consider an execution of the \TT{FS} algorithm in which the local objective function
of each agent $i$ is $f_i(x)$. Then \TT{FS} algorithm provides the following privacy guarantees:
\begin{itemize} 
\item [(P1)] Let the network graph $\MC{G}$ have a minimum degree $\geq f+1$. For any agent $i\not\in\MC{A}$, choose any feasible local objective function $g_i(x) \neq f_i(x)$. The adversary's observations in the above execution are compatible with at least one feasible problem in $\MC{F}$ in which agent $i$'s local objective function equals $g_i(x)$. In other words, the adversary cannot learn function $f_i(x)$ for $i\not\in\MC{A}$.
\item [(P2)] Let the network graph $\MC{G}$ have vertex connectivity $\geq f+1$. For each $\MC{I}\subset\MC{V}-\MC{A}$, choose a feasible local objective function $g_i(x)\neq f_i(x)$ for each $i\in\MC{I}$. The adversary's observations in the above execution are compatible at least one feasible problem in $\MC{F}$ wherein, for $i\in \MC{I}$, agent $i$'s local objective function is $g_i(x)$. In other words, the adversary cannot learn $\sum_{i \in \MC{I}}f_i(x)$.
\end{itemize}

\end{theorem}
\noindent The proof for property (P2) in Theorem~\ref{Th:TPriv-2} is sketched in Section~\ref{Sec:Privacy:Proofs} and detailed in Appendix~\ref{Sec:Appendix-T2Proof}. Property (P1) can be proved similarly. 

~

\noindent \textbf{Convergence-Privacy Trade-off: }
Addition of perturbations to the state estimates can 
improve privacy, however,
it also degrades the convergence rate.
Analogous to the finite-time analysis presented in \cite{duchi2012dual}, the theorem below
assumes $\alpha_k=1/\sqrt{k}$, and provides a convergence result for a weighted time-average
of the state estimates $\widehat{x}^j_T$ defined below.

\begin{theorem}  \label{Th:FiniteTimeRes}
Let estimates $\{x^j_k\}$ be generated by \TT{RSS-NB} or \TT{RSS-LB} with $\alpha_k = 1/\sqrt{k}$.
For each $j\in\MC{V}$, let $$\widehat{x}^j_T = \frac{\sum_{k=1}^T \alpha_k x^j_k}{\sum_{k=1}^T \alpha_k}.$$
Then,
\begin{align*}
    f(\widehat{x}^j_T) - f(x^*) = \mathcal{O}\left((1 + \Delta^2) \frac{\log(T)}{\sqrt{T}}\right).
\end{align*}

\end{theorem}%

Section \ref{Sec:Analysis} presents the proof.
The above theorem shows that the gap between the optimal function value and
function value at the time-average of state estimates ($\widehat{x}^j_k$) has
a gap that is quadratic in noise bound $\Delta$. The dependence on
time $T$ in the convergence result above is similar to that for \TT{DGD} in \cite{duchi2012dual,Nedi2015}, and is a consequence of the consensus-based local gradient method used here. The quadratic dependence on $\Delta$ is a consequence of structured randomization. Larger $\Delta$ results in slower convergence, however, would result in larger randomness in the iterates, improving privacy.

Random perturbations used in algorithms \TT{RSS-NB} and \TT{RSS-LB} cause a slowdown in convergence, however, do not introduce an error
in the outcome. This is different from $\tilde{\epsilon}$-Differential Privacy where perturbations result in slowdown in addition to an error of the order of $\mathcal{O}(1/\tilde{\epsilon}^2)$ \cite{Han2016}.



\section{Performance Analysis} \label{Sec:Analysis}

We sketch the analysis of \TT{RSS-NB} here. Analysis of \TT{RSS-LB} and \TT{FS} follows similar structure.
For brevity, only key results are presented here. Detailed proofs are available in Appendix (also \cite{gade16convsum,gade2016private}).
We often refer to the {\em state estimate} of an agent as its {\em iterate}.
Define iterate average ($\bar{x}_{k}$, at iteration $k$) and the disagreement of iterate $x^j_{k}$ agent $j$ with $\bar{x}_{k}$ as,
\begin{align}
\bar{x}_{k} &= \frac{1}{n}\sum_{j=1}^n x^j_{k}, \text{ and }  \label{Eq:AvgIterateDisagree}\\
\delta^j_{k} &= x^j_{k} - \bar{x}_{k}. \label{Eq:AvgIterateDisagree2}
\end{align}

The computation on Line 8 of \TT{RSS-NB} Algorithm~\ref{Algo:PPDOP2} can be represented using ``true-state'', denoted as $\widehat{v}^j_k$, and a perturbation, $e^j_k$,
as follows. 
\begin{eqnarray}
\widehat{v}^j_k &=& \sum_{i =1}^n B_k[j,i] x^i_{k},  \label{Eq:InfFusALG2:v}\\
e^j_k &=& \sum_{i =1}^n B_k[j,i] d^i_k, \label{Eq:InfFusALG2:v2}\\
v^j_{k}&=& \sum_{i =1}^n B_k[j,i] w^i_{k} ~=~ \widehat{v}^j_k + \alpha_k e^j_k \label{Eq:InfoFusionN:v}
\end{eqnarray}%
Since $\sum_j d^j_k = 0$ and $B_k$ is doubly stochastic, we get, $\sum_j e^j_k = 0$ (see Appendix for details).   

Similarly for \TT{RSS-LB}, computation on Line 8 of \TT{RSS-LB} Algorithm~\ref{Algo:PPDOP3} can be represented using ``true state'', denoted as $\widehat{v}^j_k$, and a perturbation, $e^j_k$,
as follows. 
\begin{eqnarray}
\widehat{v}^j_k &=& \sum_{i =1}^n B_k[j,i] x^i_{k},  \label{Eq:InfFusALG2:vLB}\\
e^j_k &=& \sum_{i =1}^n B_k[j,i] d^{i,j}_k, \label{Eq:InfFusALG2:v2LB}\\
v^j_{k}&=& \sum_{i =1}^n B_k[j,i] w^{i,j}_{k} ~=~ \widehat{v}^j_k + \alpha_k e^j_k \label{Eq:InfoFusionN:vLB}
\end{eqnarray}%
Following the construction of noise (Line 4 in Algorithm~\ref{Algo:PPDOP3}) we can show $\sum_j e^j_k = 0$ (see Appendix for details).

Now we can represent the projected gradient descent step (Line 9 of Algorithm~\ref{Algo:PPDOP2} or Line 9 of Algorithm~\ref{Algo:PPDOP3}) as,
{\begin{align}
x^j_{k+1}  = \MC{P}_{\MC{X}}\left[ \widehat{v}^j_k - \alpha_k \left( \nabla f_j(v^j_k) - e^j_k\right) \right]. \label{Eq:ProjGradPerspective}
\end{align}}%
Note that the same equation holds for both \TT{RSS-NB} and \TT{RSS-LB} where $e^j_k$ definition is defined by Eq.~\ref{Eq:InfFusALG2:v2} for \TT{RSS-NB} and Eq.~\ref{Eq:InfFusALG2:v2LB} for \TT{RSS-LB}. In the above expression, the perturbation can be viewed simply as noise in the gradient. This perspective 
is useful for the analysis.
%
Using a result from \cite{ram2010distributed} on linear convergence of product of doubly stochastic matrices, we obtain a bound
on disagreement $\|\delta^j_k\|$ in Lemma~\ref{Lem:AvgDisagreement1} below.
\begin{lemma} \label{Lem:AvgDisagreement1}
For constant $\beta<1$ and constant $\theta$ that both only depend on the network $\MC{G}$, doubly-stochastic matrices $B_k$, iterates $x^j_k$ generated by \TT{RSS-NB}, \TT{RSS-LB}, and for $k \geq 1$
\begin{align*}
\max_{j\in \MC{V}} \|\delta^j_{k+1}\| &\leq n \theta \beta^{k} \max_{i \in \MC{V}} \|x^i_1\| + 2 \alpha_k \left( L + \Delta\right) + n \theta (L+\Delta) \sum_{l=2}^k \beta^{k+1-l} \alpha_{l-1}
\end{align*}%
\end{lemma}

\noindent The proof of Lemma~\ref{Lem:AvgDisagreement1} is presented in Appendix~\ref{Sec:Appendix-L1Proof}.


Lemma \ref{Lem:AvgDisagreement1} can be used to show that the iterates maintained by the different agents
asymptotically reach consensus. 
Lemma~\ref{Lem:IterateConvRelation} below provides a bound on the distance between iterates and the optimum. 
\begin{lemma} \label{Lem:IterateConvRelation}
For iterates $x^j_k$ generated by \TT{RSS-NB}, \TT{RSS-LB}, $y \in \mathcal{X}$ and $k\geq 1$, the following holds,
\begin{align*}
&\eta_{k+1}^2 \leq \left(1 + F_k \right)\eta_k^2 - 2 \alpha_k \left(f(\bar{x}_k) - f(y) \right) + H_k, \; \; \\
& \mbox{where~} \eta_k^2 = \sum_{j=1}^n \|x^j_k - y\|^2, \quad F_k = \alpha_k N \left( \max_{j\in\MC{V}} \|\delta^j_k\| + \alpha_k \Delta\right), \text{ and}\\
&H_k = 2 \alpha_k n (L+\frac{N}{2}+\Delta) \max_{j \in \MC{V}} \|\delta^j_k\|  + \alpha_k^2 n \left(N\Delta +(L +\Delta)^2 \right).
\end{align*}%
\end{lemma}
\noindent The proof of Lemma~\ref{Lem:IterateConvRelation} is presented in Appendix~\ref{Sec:Appendix-L2Proof}.
%
The expressions in Lemma~\ref{Lem:IterateConvRelation} has the same structure as supermartingale convergence result from \cite{robbins1985convergence}. We can show that $\sum_k H_k < \infty$ and $\sum_k F_k < \infty$. Then using the result from \cite{robbins1985convergence} asymptotic convergence of the iterate average $\bar{x}_k$ to an optimum $x^*\in\MC{X}^*$ can be proved, proving Theorem~\ref{Th:ConvPPDOP2}.

~

\noindent {\bf Proof of Thoerem \ref{Th:FiniteTimeRes}:}
Next, we sketch the proof of Theorem~\ref{Th:FiniteTimeRes}, which uses
Lemma~\ref{Lem:IterateConvRelation}. The detailed proof of Theorem~\ref{Th:FiniteTimeRes} is presented in Appendix~\ref{Sec:Appendix-FiniteTRes}. As discussed earlier, Theorem ~\ref{Th:FiniteTimeRes} assumes $\alpha_k=1/\sqrt{k}$.
Recall the definition of ${\widehat{x}}_T^j$ in Theorem \ref{Th:FiniteTimeRes}.
Let $\bar{\widehat{x}}_T = \frac{1}{n}\sum_{j=1}^n {\widehat{x}}_T^j$.  
Observing that $\bar{\widehat{x}}_T$ also equals $\sum_{k=1}^T\alpha_k\bar{x}_k$ and
using the fact that $f(x)$ is convex, we get,
{\small
\begin{align*}
f(\bar{\widehat{x}}_T) - f^* \leq \frac{\sum_{k=1}^T \alpha_k f(\bar{x}_k)}{\sum_{k=1}^T \alpha_k} - f^* = \frac{\sum_{k=1}^T \alpha_k \left(f(\bar{x}_k) - f^*\right)}{\sum_{k=1}^T \alpha_k}
\end{align*}}
\noindent Lemma~\ref{Lem:IterateConvRelation} and the observation ${\small \sum_{k=1}^T \alpha_k \geq \sqrt{T}}$ yields:
{\small
\begin{align*}
f(\bar{\widehat{x}}_T) - f^* &\leq \frac{\sum_{k=1}^T \left((1+F_k)\eta^2_k - \eta^2_{k+1} + H_k\right)}{2\sum_{k=1}^T \alpha_k} \\
&\leq \frac{\eta_1^2 + \sum_{k=1}^T \left(F_k \eta^2_k + H_k\right)}{2\sqrt{T}} 
\end{align*}}

\noindent Next we bound $\sum_{k=1}^T F_k$ and $\sum_{k=1}^T H_k$.
\begin{align*}
\sum_{k=1}^T F_k &= 2N \sum_{k=1}^T \alpha_k \max_j \|\delta^j_k\| + 2N\Delta \sum_{k=1}^T \alpha_k^2 \leq 2N \sum_{k=1}^T \alpha_k \max_j \|\delta^j_k\| + 2N \Delta (\log(T) + 1) \\
\sum_{k=1}^T H_k &\leq 2n(L+N/2+\Delta) \sum_{k=1}^T \alpha_k \max_j \|\delta^j_k\| + n[(L+\Delta)^2 + N \Delta] (\log(T)+1) 
\end{align*}

\noindent Use Lemma~\ref{Lem:AvgDisagreement1} to bound $\sum_{k=1}^T \alpha_k \max_j \|\delta^j_k\|$.
\begin{align*}
f(\bar{\widehat{x}}_T) - f^* \leq \frac{C_0 + C_1\log(T) + C_{2}\log(T-1)}{\sqrt{T}}
\end{align*}

\noindent We then use the Lipschitzness of $f(x)$ to arrive at,
\begin{align*}
f(\widehat{x}^j_T) - f^* &= f(\widehat{x}^j_T) - f(\bar{\widehat{x}}_T) + f(\bar{\widehat{x}}_T) - f^* \\
&\leq L \|\widehat{x}^j_T - \bar{\widehat{x}}_T\| + \frac{C_0 + (C_1+C_2)\log(T)}{2\sqrt{T}}\qquad \mbox{where~} C_1, C_2 = \MC{O}(\Delta^2) \\
&= \MC{O}\left((1+ \Delta^2 )\frac{\log(T)}{\sqrt{T}}\right) 
\end{align*} 
\hfill $\square$

\section{Privacy with Function Sharing}
\label{Sec:Privacy:Proofs}

In this section we consider a special case of Problem~\ref{Prob:DistLearn}. Assume that all objective functions $f_i(x)$ are polynomials with degree $\leq d$. Consequently, $f(x)$ is a polynomial with $\text{deg}(f(x)) \leq d$.
We now prove property (P2) in Theorem~\ref{Th:TPriv-2}; the proof of property (P1) can be obtained similarly. 
\comment{++++++++++++
\noindent (\textit{Necessity}) We prove this by contradiction. Let us assume that $\kappa(\MC{G}) = f$. Hence, if $f$ nodes are deleted (along with their edges) the graph becomes disconnected to form components $I_1$ and $I_2$, (see Figure~\ref{Fig:NecProof1}). 

Consider that $f$ nodes, that we just deleted, were compromised by an adversary. Due to the strong assumption on adversarial knowledge, the adversary can estimate the obfuscated function at each node $\widehat{f}_i(x)$.

Agents generate correlated noise function using Eq.~\ref{Eq:SMCFunctionGen}. All the perturbation functions $s^{j,i}(x)$ shared by nodes in $I_1$ (with agents outside $I_1$) are observed by the adversary. This follows from the fact that all the neighbors of $I_1$ (nodes outside $I_1$) are compromised. Hence, the true objective function for the the subset of nodes in $I_1$ can be easily estimated by using,
 \begin{small}
 \begin{align*}
 \sum_{l \in I_1} f_{l}(x) = \sum_{l \in I_1} \widehat{f}_{l} - \left(\sum_{j = 1, i\in I_1}^{j=f} s^{j,i}(x) - \sum_{j = 1, i\in I_1}^{j=f} s^{i,j}(x) \right).
 \end{align*}%
 \end{small}%
This gives us contradiction. If $\kappa(\MC{G}) = f$, then the privacy of $I_1$ can be broken. Hence, $\kappa(\MC{G})>f$ is necessary. 
\begin{figure}[!h]
 \begin{center}
 \centerline{\includegraphics[width=0.7\columnwidth]{NecessityProof-Pic2n}}
 \caption{(Necessity) An example of topology with $\kappa(\MC{G}) = f$.}
 \label{Fig:NecProof1}
 \end{center}
\end{figure}%

\noindent (\textit{Sufficiency})
+++++++++++++++++++++++}

~

\noindent \textbf{Proof of Theorem 2 (P2): }In particular, we use a constructive approach to show that given any execution, any feasible candidate for objective functions of nodes in any subset $\MC{I}\subset\MC{V}-\MC{A}$ is compatible with the adversary's observations.


Suppose that the actual local objective function of each agent $i$ in a given execution is $f_i(x)$.  Consider subset $\MC{I}\subset\MC{V}-\MC{A}$, and consider any feasible local objective function $g_i(x)$ for each $i\in\MC{I}$.
Now, for any $i\in\MC{A}$, let $g_i(x)=f_i(x)$ (adversary observes own objective functions). Also, since the functions
are polynomials of bounded degree, it should be easy to see that,
for each $i\in\MC{V}-\MC{A}-\MC{I}$, we can find local objective functions $g_i(x)$
such that $\sum_{i\in\MC{V}-\MC{A}}g_i(x)=\sum_{i\in\MC{V}-\MC{A}}f_i(x)$, for all $x\in\MC{X}$. Thus, for the given functions $g_i(x)$ for $i\in\MC{I}$, we have
found feasible local objective functions $g_i(x)$ for all agents such that -- i) the
local objective functions of compromised agents are identical to those in the actual
execution, and ii) the sum of objective functions of ``good'' nodes is preserved $\sum_{i\in\MC{V}-\MC{A}} g_i(x) = \sum_{i\in\MC{V}-\MC{A}}f_i(x)$. 

Recall that the {\em function sharing} algorithm adds noise functions
to obtain perturbed function $\widehat{f}_i(x)$ at each agent $i\in\MC{V}$.
In particular, agent $j$ sends to each neighboring
agent $i$ a noise function, say $s^{j,i}(x)$,
and subsequently computes $\widehat{f}_i(x)$ using the noise functions it sent
to neighbors and the noise functions received from the neighbors.

When the vertex connectivity of the graph is at least $f+1$ it is easy to show that,
for local objective functions 
$\{g_1(x),g_2(x),\cdots,g_n(x)\}$ defined above, each agent $j\in\MC{V}-\MC{A}$
can select noise functions, say $t^{j,i}(x)$ for each neighbor $i$, with the following properties:
\begin{itemize}
\item For each $j\in\MC{V}-\MC{A}$ and neighbor $i$ of $j$ such that $i\in\MC{A}$, $t^{j,i}(x)=s^{j,i}(x)$ for all $x\in\MC{X}$. That is,
the noise functions exchanged with agents in $\MC{A}$ are unchanged.
\item For each $j\in\MC{V}-\MC{A}$, 
 $$g_j(x) + \sum_{i \in \MC{N}_j} t^{i,j}(x) - \sum_{i \in \MC{N}_j} t^{j,i}(x) = \widehat{f}_j(x).$$
That is, the obfuscated function of each agent in $\MC{V}-\MC{A}$ remains
the same as that in the original execution.
\end{itemize}
Due to the above two properties, the observations of the adversary in the above
execution will be identical to those in the original execution. Thus,
the adversary cannot distinguish between the two executions. This, in turn, implies property (P2) in Theorem \ref{Th:TPriv-2}.
Property (P1) can be proved similarly. \hfill $\square$

\section{Experimental Results}

We now provide some experimental results for \TT{RSS-NB} and \TT{RSS-LB} algorithms. We present two sets of experiments. First, we show that \TT{RSS-NB} and \TT{RSS-LB} correctly solve distributed optimization of polynomial objective functions. Next, we apply our algorithms in the context of machine learning for handwritten digit classification (using MNIST dataset) and document classification (using Reuters dataset). 

~
	
\noindent \textbf{Polynomial Optimization: } We solve polynomial optimization on a network of 5 agents that form
a cycle. The objective functions of the 5 agents are chosen as $f_1(x) = x^2$, $f_2(x) = x^4$, $f_3(x) = x^2+x^4$, $f_4(x) = x^2 + 0.5x^4$, and $f_5(x) = 0.5 x^2 + x^4$. The aggregate function $f(x) = 2.5(x^2+x^4)$. We consider $\MC{X} = [-30,30]$. Simulation results in Figure~\ref{fig:polyres} show that the two algorithms converge to the optimum $x^* = 0$ for two different values of $\Delta$. Large $\Delta$ results in larger perturbations and the convergence is slower. For smaller $\Delta$, as expected, the performance of both \TT{RSS-NB} and \TT{RSS-LB} is closer to \TT{DGD}.

\begin{figure}[!t]
    \begin{centering}
       \includegraphics[width=0.49\linewidth, keepaspectratio]{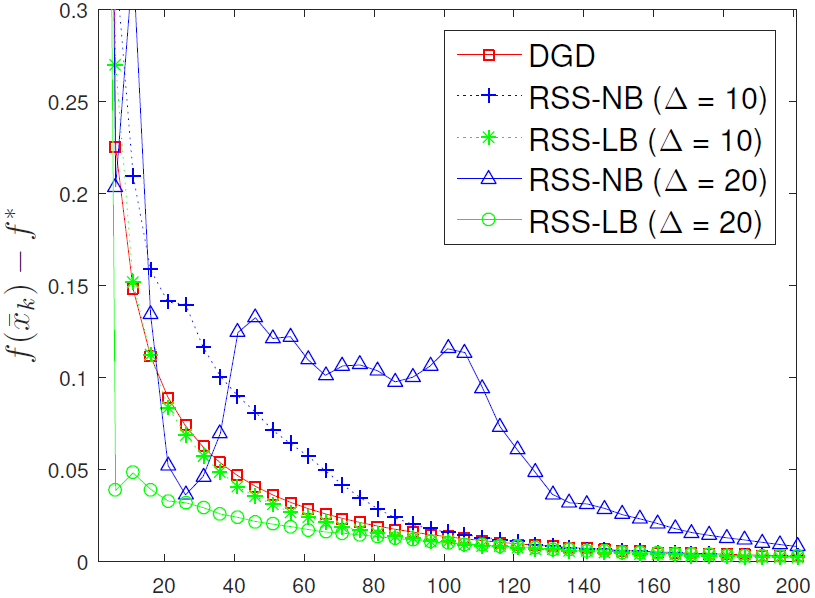}
       \caption{Function suboptimality v/s iterations.}
       \label{fig:polyres}
   \end{centering}
\end{figure}%

\begin{figure}[!t]
    \begin{centering}
       \includegraphics[width=1.0\linewidth, keepaspectratio]{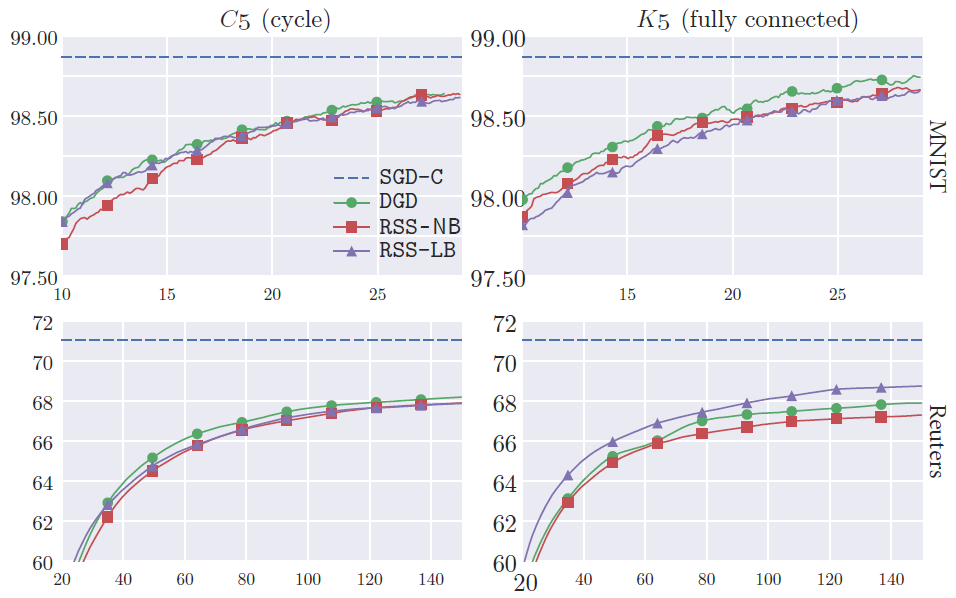}
       \caption{Accuracy for MNIST (top row) and Reuters (bottom row) with $C_5$ (left) and $K_5$ (right) topologies. $Y$-axis is the testing accuracy in \% and $X$-axis is epochs. \TT{SGD-C} is a centralized solution.}
       \label{fig:accuracy}
   \end{centering}
\end{figure}%

~

\noindent \textbf{Machine Learning: }
We consider two classification problems. We use a deep neural network \cite{shokri2015privacy} for digit recognition using the MNIST dataset \cite{lecun-mnisthandwrittendigit-2010} and regularized logistic regression for the Reuters dataset \cite{DBLP:journals/corr/Ludwig16}. We use two graph topologies: a cycle
of 5 agents (namely, $C_5$) and a complete graph of 5 agents (namely, $K_5$). 

Due to the high overhead of computing gradients on full datasets, we evaluate versions of the different algorithm that are
adapted to perform 
stochastic gradient descent (SGD) on minibatches of local, non-overlapping datasets, to solve the distributed machine learning problem. Also, by performing consensus only every 40 gradient steps of gradient descent, we decrease network overhead, while still retaining the accuracy. Figure~\ref{fig:accuracy} shows convergence results for \TT{DGD}, \TT{RSS-NB} and \TT{RSS-LB}, and a centralize algorithms \TT{SGD-C}, which demonstrate that our algorithms can achieve high accuracy despite the introduction of perturbances
in state estimates.

\begin{itemize}
\item \textbf{MNIST: }Our algorithms converge quickly despite the deep learning problem being non-convex and despite using stochastic gradient descent. 
\item
\textbf{Reuters: }We achieve testing accuracy comparable to a centralized solution \TT{SGD-C}. \TT{DGD} works best followed by \TT{RSS-NB} and \TT{RSS-LB} for cycle topology. 
\end{itemize}

\section{Conclusion}

In this paper, we develop and analyze iterative distributed optimization algorithms \TT{RSS-NB}, \TT{RSS-LB} and \TT{FS} that exploit {\em structured randomness} to improve privacy, while maintaining accuracy. We prove convergence and develop trade-off between convergence rate and the bound on perturbation. We provide
claims of privacy for the \TT{FS} algorithm, which is a special case of \TT{RSS-NB}. We apply versions of \TT{RSS-NB} and \TT{RSS-LB} to distributed machine learning, and evaluating their effectiveness for training with MNIST and Reuters datasets.

\bibliography{Central.bib}
\bibliographystyle{ieeetr}
\appendix

\section{Notation Summary}
Summary of symbols and constants used in the analysis is presented in Table~\ref{Tab:Notation}

\begin{table*}[!h]
  \centering
  \begin{tabular}{|c|l|}
  \hline
    Symbol & Description  \\ \hline
    $\MC{V}$ & Set of Agents \\ \hline
    $\MC{E}$ & Set of Communication Links (Edges) \\ \hline
    $\MC{G(V,E)}$ & Graph with Nodes $\MC{V}$ and Edges $\MC{E}$\\ \hline    
    $n$ & Number of Agents $n = |\MC{V}|$ \\ \hline
    $\MC{N}_j$ & Neighborhood set of agent $j$ \\ \hline    
    $f(x)$ & Objective function \\ \hline
    $f_i(x)$ & Objective function associated with Agent $i$ \\ \hline
    $D$ & Dimension of the optimization problem (size of vector $x$) \\ \hline
    $\MC{X}$ & Feasible Set (Model Parameters) \\ \hline
    $\MC{X}^*$ &Optimal Set (Model Parameters) \\ \hline
    $x^j_k$ & Estimate of iterate (state) with agent $j$ at time $k$ \\ \hline
    $\bar{x}_k$ & Average iterate (over the network) at time $k$ \\ \hline
    $\delta^j_k$ & Disagreement between agent $j$ and iterate-average, $\delta^j_k = x^j_k - \bar{x}_k$ \\ \hline
    $L$ & Gradient bound, $\|\nabla f_i(x)\| \leq L$ \\ \hline
    $N$ & Lipschitz constant for $\nabla f_i(x)$ \\ \hline
    $B_k$ & Doubly stochastic matrix \\ \hline
    $\rho$ & Smallest non-zero entry in matrix $B_k$ \\ \hline
    $\alpha_k$ & Learning step-size \\ \hline
    $d^j_k$ & Correlated random perturbation \\ \hline
    $w^j_k$ & Perturbed local model estimate \\ \hline
    $\Delta$ & Random perturbation bound, $\|d^j_k\| \leq \Delta$  \\ \hline
    $v^j_k$ & Fused (perturbed) model estimate, Eq.~\ref{Eq:InfoFusionN:v} \\ \hline
    $\hat{v}^j_k$ & Fused (true) model estimate, Eq.~\ref{Eq:InfFusALG2:v} \\ \hline
    $e^j_k$ & Fused perturbation, Eq.~\ref{Eq:InfFusALG2:v2} \\ \hline
    $\Phi(k,s)$ & Transition Matrix \\ \hline
    $\theta$ & Transition matrix convergence parameter, $\theta = (1 - \frac{\rho}{4n^2})^{-2}$ \\ \hline
    $\beta$ & Transition matrix convergence parameter, $\beta = (1 - \frac{\rho}{4n^2})$ \\ \hline    
    $\eta_k^2$ & Distance between iterate and optimum, $\eta_k^2 = \sum_{j=1}^n \|x^j_k - y\|^2$ \\ \hline 
    $\xi_k^2$ & Distance between fused iterate and optimum, $\xi_k^2 = \sum_{j=1}^n \|\hat{v}^j_k - y\|^2 $ \\ \hline
    $F_k$ & $F_k = \alpha_k N \left( \max_{j \in \MC{V}} \|\delta^j_k\| + \alpha_k \Delta\right)$  \\ \hline 
    $H_k$ & $H_k = 2 \alpha_k n \max_{j \in \MC{V}} \|\delta^j_k\| (L+N/2+\Delta) +\alpha_k^2 n \left[N\Delta +(L +\Delta)^2 \right]$      \\ \hline     
  \end{tabular}
  \caption{Parameters and Constants}
  \label{Tab:Notation}
\end{table*}

We will analyze both \TT{RSS-NB} and \TT{RSS-LB} simultaneously. Before we move on to the proofs, we first define key bounds on the error for both the algorithms. We first establish the boundedness of error term in Eq.~\ref{Eq:InfFusALG2:v2} (for \TT{RSS-NB}) and in Eq.~\ref{Eq:InfoFusionN:vLB} (for \TT{RSS-LB}). We also show that the error adds to zero over the network. These relationships will be critical for proving the convergence of our algorithms.  

\begin{itemize} [leftmargin=*]
\item \TT{RSS-NB}
\begin{itemize}
\item {\bf Boundedness ($\|e^j_k\| \leq \Delta$):} The noise perturbations $s^{i,j}_k$ are bounded by $\Delta/2n$. It follows that, $d^j_k = \sum_{i \in \MC{N}_j} s^{i,j}_k - \sum_{i \in \MC{N}_j} s^{j,i}_k$ will also be bounded, $\|d^j_k\| \leq \Delta$ since node $j$ may have at most $n-1$ neighbors. Following Eq.~\ref{Eq:InfFusALG2:v2}, we can show, $\|e^j_k\| = \|\sum_{i \in \MC{N}_j} B_k[j,i]d^i_k\| \leq \sum_{i \in \MC{N}_j} B_k[j,i] \|d^i_k\|$ since $B_k[j,i]$ is non-negative.	
\begin{align}
\|e^j_k\| \leq \sum_{i =1}^n B_k[j,i] \|d^i_k\| \leq \Delta \sum_{i =1}^n B_k[j,i] = \Delta. \label{Eq:NB-Bounded}
\end{align}
\item {\bf Aggregate Randomness ($\sum_j e^j_k = 0$):} We prove this using Eq.~\ref{Eq:InfFusALG2:v2}, followed by the fact that $B_k$ matrix is a doubly stochastic matrix. 
\begin{align}
\sum_{j=1}^n e^j_k = \sum_{j=1}^n \left( \sum_{i \in \MC{N}_j} B_k[j,i] d^i_k \right) = \sum_{i=1}^n \left( \sum_{j=1}^n B_k[j,i]\right) d^i_k = \sum_{i=1}^n d^i_k = 0. \qquad \qquad \text{(from Eq.~\ref{Eq:RSSNB-NoiseChar})} \label{Eq:NB-Agg}
\end{align}
\end{itemize}

\item \TT{RSS-LB}
\begin{itemize}
\item {\bf Boundedness ($\|e^j_k\| \leq \Delta$):} The noise perturbations $d^{i,j}_k$ are bounded by $\Delta$ (by construction) i.e. $\|d^{i,j}_k\| \leq \Delta$. Following Eq.~\ref{Eq:InfFusALG2:v2LB}, we have, $\|e^j_k\| = \|\sum_{i \in \MC{N}_j} B_k[j,i]d^{i,j}_k\| \leq \sum_{i \in \MC{N}_j} B_k[j,i] \|d^{i,j}_k\|$ since $B_k[j,i]$ is non-negative and $B_k$ is doubly stochastic.	
\begin{align}
\|e^j_k\| \leq \sum_{i =1}^n B_k[j,i] \|d^{i,j}_k\| \leq \Delta \sum_{i =1}^n B_k[j,i] = \Delta. \label{Eq:LB-Bounded}
\end{align}
\item {\bf Aggregate Randomness ($\sum_j e^j_k = 0$):} We prove this using Eq.~\ref{Eq:InfFusALG2:v2LB}, and Eq.~\ref{eq:lb}. 
\begin{align}
\sum_{j=1}^n e^j_k = \sum_{j=1}^n \left( \sum_{i \in \MC{N}_j} B_k[j,i] d^{i,j}_k \right) = 0. \qquad \qquad \text{(from Eq.~\ref{eq:lb})}\label{Eq:LB-Agg}
\end{align}
\end{itemize}
\end{itemize}

Next, we prove an important relationship between the iterate-average and average of ``true-state'' defined in Eq.~\ref{Eq:InfoFusionN:v} (for \TT{RSS-NB}) and in Eq.~\ref{Eq:InfoFusionN:vLB} (for \TT{RSS-LB}). 

\begin{align}
\bar{v}_k \triangleq \frac{1}{n} \sum_{j=1}^n \hat{v}^j_k &= \frac{1}{n} \sum_{j=1}^n \left( \sum_{i=1}^n B_k[j,i] x^i_k \right) \nonumber \\
&=\frac{1}{n} \sum_{i=1}^n \left( \sum_{j=1}^n B_k[j,i] \right) x^i_k      \nonumber \\
&= \frac{1}{n} \sum_{i=1}^n x^i_k  = \bar{x}_k  \label{Eq:IterAverPreserved} && \sum_{j=1}^n B_k[j,i] = 1, B_k \text{ is column stochastic} 
\end{align}

\noindent We have showed that the iterate-average stays preserved under convex averaging. This proof is adapted from \cite{nedic2009distributed}.

We define transition matrix, $\Phi(k,s)$, as the product of doubly stochastic weight matrices $B_k$, $$\Phi(k,s) = B_k B_{k-1} \ldots B_{s+1} B_s, \quad (\forall \ k \geq s > 0).$$ 

We first note two important results from literature. The first result relates to convergence of non-negative sequences (Lemma~\ref{Lem:Ram}) and the second result describes the linear convergence of transition matrix to $\frac{1}{n} \mathbb{1}\mathbb{1}^T$ (Lemma~\ref{Lem:NedichCorollary}).

\setcounter{lem_count}{2}
\begin{lemma} [Lemma 3.1, \cite{ram2010distributed}] \label{Lem:Ram}
Let $\{\zeta_k\}$ be a non-negative scalar sequence. If $\sum_{k=0}^\infty \zeta_k < \infty$ and $0 < \beta < 1$, then $\sum_{k=0}^\infty \left( \sum_{j=0}^k \beta^{k-j} \zeta_j \right) < \infty$.
\end{lemma}

\begin{lemma} [Corollary 1, \cite{nedic2008distributed}] \label{Lem:NedichCorollary}
Let the graph $\MC{G}$ be connected, then,
\begin{enumerate}[leftmargin=*,topsep=0pt,itemsep=0pt]
\item $\lim_{k \rightarrow \infty} \Phi(k,s) = \frac{1}{n} \mathbb{1}\mathbb{1}^T$ for all $s > 0$.
\item $|\Phi(k,s)[i,j] - \frac{1}{n}| \leq \theta \beta^{k-s+1}$ for all $k \geq s > 0$, where $\theta = (1 - \frac{\rho}{4n^2})^{-2}$ and $\beta = (1 - \frac{\rho}{4n^2})$.
\end{enumerate}
\end{lemma}


\noindent The well known non-expansive property (cf.~\cite{bertsekas2003convex}) of Euclidean projection onto a non-empty, closed, convex set $\mathcal{X}$, is represented by the following inequality, $\forall \; x, y \in \mathbb{R}^D$, 
\begin{align}
    \| \mathcal{P}_\mathcal{X}[x] - \mathcal{P}_\mathcal{X}[y] \| \leq \|x - y\|. \label{Eq:EUCPROJINEQUAL}
\end{align}

We state the deterministic version of a known result on the convergence of non-negative almost supermartinagales. 

\begin{lemma} {\normalfont \cite{robbins1985convergence}}. \label{Lem:RobSiegConv}
Let $\{F_k\}$, $\{E_k\}, \{G_k\}$ and $\{H_k\}$, be non-negative, real sequences. Assume that $\sum_{k=0}^\infty F_k < \infty$, and $\sum_{k=0}^\infty H_k < \infty$ and $$E_{k+1} \leq  (1+F_k) E_k - G_k + H_k.$$ Then, the sequence \{$E_k$\} converges to a non-negative real number and $\sum_{k=0}^\infty G_k < \infty$.
\end{lemma}

The proofs in this report follow the structure similar to \cite{ram2010distributed}, with the key difference being in proof for Lemma~\ref{Lem:IterateConvRelation} and Theorem~\ref{Th:ConvPPDOP2}.

\subsection{Proof of Lemma~\ref{Lem:AvgDisagreement1} (Disagreement Lemma)} \label{Sec:Appendix-L1Proof}
\begin{proof}
Define for all $j \in \mathcal{V}$ and all $k$,
\begin{align}
    z^j_{k+1} &= x^j_{k+1} - \sum_{i=1}^n B_k[j,i] x^i_k \\
    \implies x^j_{k+1} &=  z^j_{k+1} + \sum_{i=1}^n B_k[j,i] x^i_k
\end{align}
We then unroll the iterations to get $x^j_{k+1}$ as a function of $z^j_{k+1}$, and $x^i_k$ and doubly stochastic weight matrix at current and previous iteration. 
\begin{align*}
x^j_{k+1} &= z^j_{k+1} + \sum_{i=1}^n \left[ B_k[j,i] \left( z^i_k + \sum_{l=1}^n B_{k-1}[i,l] x^l_{k-1}\right) \right]
\end{align*}
We perform the above mentioned unrolling successively and use the definition of transition matrix $\Phi(k,s)$, 
\begin{align}
    x^j_{k+1} &= z^j_{k+1} + \sum_{i=1}^n \Phi(k,1)[j,i] x^i_1 + \sum_{l=2}^k \left[\sum_{i=1}^n \Phi(k,l)[j,i] z^i_{l}\right]. \label{Eq:ItAv1}
\end{align}%
\noindent Note that $\Phi(1,1) = B_1$. We verify the expression for $k = 1$, and we get the relationship $x^j_{2} = z^j_{2} +\sum_{i=1}^n \Phi(1,1)[j,i] x^i_1$. 

We can write the relation for iterate average, $\bar{x}_k$, and use doubly stochastic nature of $B_k$ to get,
\begin{small}
 \begin{align}
    \bar{x}_{k+1} = \frac{1}{n} \sum_{j=1}^n x^j_{k+1} &= \frac{1}{n} \sum_{j=1}^n \left( \sum_{i=1}^n B_k[j,i] x^i_k + z^j_{k+1} \right) = \frac{1}{n} \left( \sum_{i=1}^n \left( \sum_{j=1}^n B_k[j,i]\right) x^i_k + \sum_{j=1}^n z^j_{k+1}\right), \nonumber \\
    &= \bar{x}_k + \frac{1}{n} \sum_{j=1}^n z^j_{k+1} = \bar{x}_1 + \frac{1}{n} \sum_{l=2}^{k+1} \sum_{j=1}^n z^j_{l} \label{Eq:ItAv2}.
\end{align}
\end{small}

\noindent Using relations for $\bar{x}_{k+1}$ (Eq.~\ref{Eq:ItAv2}) and $x^j_{k+1}$ (Eq.~\ref{Eq:ItAv1}) to get an expression for the disagreement. We further use the property of norm $\|\sum a \| \leq \sum \|a\|, \; \forall a$, to get,
\begin{align}
    \|x^j_{k+1} - \bar{x}_{k+1}\| &\leq \sum_{i=1}^n \left\vert \frac{1}{n} - \Phi(k,1)[j,i] \right\vert \|x^i_1\| + \sum_{l=2}^{k} \sum_{j=1}^n \left \vert \frac{1}{n} - \Phi(k,l)[j,i]\right\vert \|z^j_l\| + \|z^j_{k+1}\| + \frac{1}{n} \sum_{j=1}^n \|z^j_{k+1}\|. \label{Eq:ItAv3}
\end{align}

\noindent We use Lemma~\ref{Lem:NedichCorollary} to bound terms of type $\left\vert \frac{1}{n} - \Phi(k,l)[j,i] \right\vert$ and $\max_{i\in\MC{V}} \|x^i_1\|$ to bound $\|x^i_1\|$, to get, 
\begin{align}
    \|x^j_{k+1} - \bar{x}_{k+1}\| &\leq n \theta \beta^{k} \max_{i\in\MC{V}} \|x^i_1\| + \theta \sum_{l=2}^k \beta^{k+1-l} \sum_{i=1}^n \|z^i_l\| + \|z^j_{k+1}\| + \frac{1}{n} \sum_{j=1}^n \|z^j_{k+1}\|. \label{Eq:ItAv4}
\end{align}
    
\noindent We now bound each of the norms $\|z^j_k\|$, using the fact that $v^j_k \in \mathcal{X}$, the non-expansive property of projection operator (see Eq.~\ref{Eq:EUCPROJINEQUAL}), and Assumption~\ref{Asmp:GradientCond},
\begin{align}
    \|z^j_{k+1}\| &= \|\mathcal{P}_{\mathcal{X}} [\hat{v}^j_k - \alpha_k \left( \nabla {f}_j(v^j_k) - e^j_k \right)] - \hat{v}^j_k\| && \text{Projected Descent for \TT{RSS-NB} and \TT{RSS-LB}, Eq.~\ref{Eq:ProjGradPerspective}}\nonumber \\
    &\leq \alpha_k \|\nabla {f}_j(v^j_k) - e^j_k\| \nonumber \\
    &\leq \alpha_k \left(\|\nabla {f}_j(v^j_k)\| + \| e^j_k\|\right) \nonumber && \text{Triangle inequality}\\
    &\leq \alpha_k \left(L + \Delta \right) && \|e^j_k\|\leq \Delta (\text{Eq.~\ref{Eq:NB-Bounded} for \TT{RSS-NB} and Eq.~\ref{Eq:LB-Bounded} for \TT{RSS-LB}}) \label{Eq:DisLemma1}
\end{align}

\noindent Note that we used the boundedness of perturbation $e^j_k$ to obtain the above relation. We show the boundedness earlier in the report: Eq:~\ref{Eq:NB-Bounded} for \TT{RSS-NB} and Eq.~\ref{Eq:LB-Bounded} for \TT{RSS-LB}. Next, recall the definition of $\delta^j_k$ from Eq.~\ref{Eq:AvgIterateDisagree2}. Combining Eq.~\ref{Eq:ItAv4} and Eq.~\ref{Eq:DisLemma1},
\begin{align}
    \max_{j \in \MC{V}} \|\delta^j_{k+1}\| = \max_{j \in \MC{V}} \|x^j_{k+1} - \bar{x}_{k+1}\|  & \leq n \theta \beta^{k} \max_{i\in\MC{V}} \|x^i_1\| + n \theta (L+\Delta) \sum_{l=2}^k \beta^{k+1-l} \alpha_{l-1} + 2 \alpha_k \left( L + \Delta \right) \nonumber
\end{align}
\end{proof}

\subsection{Proof of Lemma~\ref{Lem:IterateConvRelation} (Iterate Lemma)}  \label{Sec:Appendix-L2Proof}

\begin{proof}

Recall the relationships established in the appendix for our algorithms \TT{RSS-NB}  and \TT{RSS-LB}. 

\begin{itemize}
\item \TT{RSS-NB} (Eq.~\ref{Eq:NB-Bounded}, \ref{Eq:NB-Agg})
\begin{itemize}
\item { Boundedness $\|e^j_k\| \leq \Delta$} 
\item { Aggregate Randomness $\sum_j e^j_k = 0$} 
\end{itemize}

\item \TT{RSS-LB} (Eq.~\ref{Eq:LB-Bounded}, \ref{Eq:LB-Agg})
\begin{itemize}
\item { Boundedness $\|e^j_k\| \leq \Delta$} 
\item { Aggregate Randomness $\sum_j e^j_k = 0$} 
\end{itemize}
\end{itemize}

\noindent In summary, for both \TT{RSS-NB} and \TT{RSS-LB}, we have showed that $\|e^j_k\| \leq \Delta$ and $\sum_j e^j_k = 0$. These relationships are used for proving Lemma~\ref{Lem:IterateConvRelation} and thereby proving convergence. They allow us to perform unified analysis of both our algorithms despite the inherent differences ioin \TT{RSS-NB} and \TT{RSS-LB}.

~

To simplify analysis, we adopt the following notation, 
\begin{align}
\eta_k^2 &= \sum_{j=1}^n \|x^j_k - y\|^2, \label{Eq:BDH1}\\
\xi_k^2 &= \sum_{j=1}^n \|\hat{v}^j_k - y\|^2. \label{Eq:BDH2}
\end{align}
Note, $\eta_k$ and $\xi_k$ are both functions of $y$, however for simplicity we do not explicitly show this dependence.


Note, $\mathcal{P}_\mathcal{X}[y] = y$ for all $y \in \mathcal{X}$. Using the non-expansive property of the projection operator (Eq.~\ref{Eq:EUCPROJINEQUAL}), and the projected gradient descent expression in Eq.~\ref{Eq:ProjGradPerspective} we get, 
\begin{align}
    \|x^j_{k+1} - y\|^2 &= \|\mathcal{P}_{\mathcal{X}}\left[\hat{v}^j_{k} - \alpha_k \left(\nabla {f}_j(v^j_{k}) - e^j_k \right) \right] - y\|^2   \nonumber \\ 
    & \leq \|\hat{v}^j_{k} - \alpha_k \left(\nabla {f}_j(v^j_{k}) - e^j_k \right) - y\|^2 \label{Eq:ItLemma1} \\
    & = \|\hat{v}^j_k - y\|^2 + \alpha_k^2 \|\nabla {f}_j(v^j_k) - e^j_k\|^2 - 2 \alpha_k \left( \nabla {f}_j(v^j_k) - e^j_k \right)^T (\hat{v}^j_k - y) \label{Eq:PGRelation1}
\end{align}%

\noindent Now we add the inequalities Eq.~\ref{Eq:PGRelation1} for all agents $j = 1, 2, \ldots, n$ followed by using expressions for $\eta_k$ and $\xi_k$ (Eq.~\ref{Eq:BDH1}, Eq.~\ref{Eq:BDH2}). Next we use the boundedness of gradients (Assumption~\ref{Asmp:GradientCond}) and perturbations (Eq.~\ref{Eq:NB-Bounded} for \TT{RSS-NB} or Eq.~\ref{Eq:LB-Bounded} for \TT{RSS-LB}), to get the following inequality,
\begin{align}
	\eta_{k+1}^2 &\leq \xi_k^2 + \sum_{j=1}^n \alpha_k^2 \|\nabla {f}_j(v^j_k) - e^j_k\|^2 - 2 \alpha_k \sum_{j=1}^n \left( \nabla {f}_j(v^j_k) - e^j_k \right)^T (\hat{v}^j_k - y) \nonumber \\
				 &\leq \xi_k^2 + \sum_{j=1}^n \alpha_k^2 (\|\nabla {f}_j(v^j_k)\| + \|e^j_k\|)^2 - 2 \alpha_k \sum_{j=1}^n \left( \nabla {f}_j(v^j_k) - e^j_k \right)^T (\hat{v}^j_k - y) \qquad \text{Triangle Inequality}\nonumber \\
    \eta_{k+1}^2 &\leq \xi_k^2 + \alpha_k^2 n (L +\Delta)^2 - 2 \alpha_k \sum_{j=1}^n (\nabla {f}_j(v^j_k) - e^j_k)^T (\hat{v}^j_k - y)  \quad \text{Assumption~\ref{Asmp:GradientCond} and Eq.~\ref{Eq:NB-Bounded} or Eq.~\ref{Eq:LB-Bounded}} \label{Eq:IterateRelation1}
\end{align}


We use consensus relationship used for information fusion (Eq.~\ref{Eq:InfFusALG2:v} for \TT{RSS-NB}, or Eq.~\ref{Eq:InfFusALG2:vLB} for \TT{RSS-LB}). 
We know that in D-dimension the consensus step can be rewritten using Kronecker product of D-dimension identity matrix ($I_D$) and the doubly stochastic weight matrix ($B_k$) \cite{fax2004information}. Consider the following notation of vectors. We use bold font to denote a vector that is stacked by its coordinates. As an example, consider three vectors in $\mathbb{R}^3$ given by ${a} = [a_x, \ a_y, \ a_z]^T$, ${b} = [b_x, \ b_y, \ b_z]^T$, ${c} = [c_x, \ c_y, \ c_z]^T$. Let $\mathbf{a}$ be a vector of $a$, $b$ and $c$ stacked by coordinates, then it is defined as $\mathbf{{a}} = [a_x, \ b_x, \ c_x, \ a_y, \ b_y, \ c_y, \ a_z, \ b_z, \ c_z]^T$. Similarly we can write stacked model parameter vector as, $\mathbf{{x}}_{k} = [x^1_{k}[1], x^2_{k}[1], \ldots, x^n_{k}[1], x^1_{k}[2], x^2_{k}[2], \ldots, x^n_{k}[2], \ldots, x^1_{k}[D], \ldots, x^n_{k}[D]]^T.$ Next, we write the consensus term using the new notation and Kronecker products and compare norms of both sides (2-norm),%
\begin{align}    
	\mathbf{\hat{v}}_{k} &= (I_D \otimes B_k) \mathbf{x}_{k}  \label{Eq:E5} \\
    \mathbf{\hat{v}}_{k} - \mathbf{y} &= (I_D \otimes B_k) (\mathbf{x}_{k} - \mathbf{y}) \nonumber \\
    \|\mathbf{\hat{v}}_{k} - \mathbf{y}\|_2^2 &= \|(I_D \otimes B_k) (\mathbf{x}_{k} - \mathbf{y})\|_2^2  \nonumber \\
    &\leq \|(I_D \otimes B_k)\|_2^2 \|(\mathbf{x}_{k} - \mathbf{y})\|_2^2
\end{align}

\noindent We use the property of eigenvalues of Kronecker product of matrices. 
The eigenvalues of $I_D \otimes B_k$ are essentially $D$ copies of eigenvalues of $B_k$. Since $B_k$ is a doubly stochastic matrix, its eigenvalues are upper bounded by 1. Recall that $\|A\|_2 = \sqrt{\lambda_{\max} (A^\dagger A)}$ where $A^\dagger$ represents the conjugate transpose of matrix $A$ and $\lambda_{\max}$ represents the maximum eigenvalue. Observe that $I_D \otimes B_k$ is a doubly stochastic matrix and $(I_D \otimes B_k)^\dagger(I_D \otimes B_k)$ is also doubly stochastic matrix since product of two doubly stochastic matrices is also doubly stochastic. Clearly, $\|(I_D \otimes B_k)\|_2^2 = \lambda_{\max}((I_D \otimes B_k)^\dagger(I_D \otimes B_k)) \leq 1$.
 \footnote{An alternate way to prove this inequality would be to follow the same process used to prove Eq.~\ref{Eq:AsideEq1} except that we start with squared terms and use the doubly-stochasticity of $B_k$. Detailed derivation in Appendix~\ref{Sec:Appendix-ConProof}}
\begin{align}
    \xi_k^2 = \|\mathbf{\hat{v}}_{k} - \mathbf{y}\|_2^2 \leq \|(\mathbf{x}_{k} - \mathbf{y})\|_2^2 = \eta_k^2 \label{Eq:E6}
\end{align}

\noindent Merging the inequalities in Eq.~\ref{Eq:E6} and Eq.~\ref{Eq:IterateRelation1}, we get,
\begin{align}
    \eta_{k+1}^2 &\leq \eta_k^2 + \alpha_k^2 n (L + \Delta)^2 \underbrace{ - 2 \alpha_k \sum_{j=1}^n (\nabla {f}_j(v^j_k) - e^j_k)^T (\hat{v}^j_k - y)}_{\Lambda}. \label{Eq:IterateRelation2}    
\end{align}%

\noindent Typically, at this step one would use convexity of ${f}_j(x)$ to simplify the term $\Lambda$ in Eq.~\ref{Eq:IterateRelation2}. However, since the gradient of ${f}_j(x)$ is perturbed by noise $e^j_k$, and we need to follow a few more steps before we arrive at the iterate lemma.  

\noindent Consider the fused state iterates $\hat{v}^j_k$, the average $\bar{v}_k \triangleq (1/n) \sum_{j=1}^n \hat{v}^j_k$ and the deviation of iterate from the average, 
\begin{align} 
q^j_k = \hat{v}^j_k - \bar{v}_k. \label{Eq:FusedDisagreement}
\end{align}
We now derive a simple inequality here that will be used later. We use the fact that $\bar{x}_k = \bar{v}_k$ proved earlier in the appendix (Eq.~\ref{Eq:IterAverPreserved}).
\begin{align}
\|q^j_k\| &= \|\hat{v}^j_k - \bar{v}_k\| = \|\sum_{i=1}^n B_k[j,i]x^i_k - \bar{v}_k\| &&  \text{from Eq.~\ref{Eq:InfFusALG2:v} (\TT{RSS-NB}) or Eq.~\ref{Eq:InfFusALG2:vLB} (\TT{RSS-LB})}\nonumber \\
&= \|\sum_{i=1}^n B_k[j,i]x^i_k - \bar{x}_k\| &&  \bar{x}_k = \bar{v}_k, \text{Eq.~\ref{Eq:IterAverPreserved}} \nonumber \\
&\leq \sum_{i=1}^n B_k[j,i] \|x^i_k - \bar{x}_k\| \nonumber \\
&\leq \left( \sum_{i=1}^n B_k[j,i] \right) \max_{i\in\MC{V}} \|x^i_k -\bar{x}_k\| \nonumber \\
&\leq \max_{j \in \MC{V}} \|\delta^j_k\| && \text{Eq.~\ref{Eq:AvgIterateDisagree2}}  \label{Eq:qkinequality}
\end{align}
Note that similarly, we can derive another inequality that will be used later. 
\begin{align}
\sum_{j=1}^n \|\hat{v}^j_k - y\| &= \sum_{j=1}^n \|\sum_{i=1}^n B_k[j,i]x^i_k - y\|  &&  \text{from Eq.~\ref{Eq:InfFusALG2:v} (\TT{RSS-NB}) or Eq.~\ref{Eq:InfFusALG2:vLB} (\TT{RSS-LB})}\nonumber \\
&= \sum_{j=1}^n \|\sum_{i=1}^n B_k[j,i] \left( x^i_k - y \right)\| &&  \text{$B_k$ is row stochastic} \nonumber \\
&\leq \sum_{j=1}^n \sum_{i=1}^n B_k[j,i] \|x^i_k - y\|  \nonumber \\
&= \sum_{i=1}^n \left( \sum_{j=1}^n B_k[j,i] \right) \|x^i_k -y\| \nonumber \\
&\leq \sum_{i=1}^n \|x^i_k - y\| &&  \text{$B_k$ is column stochastic} \label{Eq:AsideEq1}
\end{align}

\noindent We use gradient Lipschiztness assumption and write the following relation,
\begin{align}
    &\nabla f_j (v^j_k) = \nabla {f}_j(\bar{v}_k) + l^j_k, \label{Eq:UnrollGradient1} 
\end{align}
where, $l^j_k$ is the (vector) difference between gradient computed at $v^j_k$ (i.e. $\nabla f_j(v^j_k)$) and the gradient computed at $\bar{v}_k$ (i.e. $\nabla f_j(\bar{v}_k)$). Next, we bound the vector $l^j_k$, using Lipschitzness of gradients,
\begin{align}
\max_{j \in \MC{V}} \|l^j_k\| &= \max_{j \in \MC{V}} \|\nabla f_j (v^j_k) - \nabla {f}_j(\bar{v}_k)\|, \nonumber \\
&\leq \max_{j \in \MC{V}} N \|{v}^j_k - \bar{v}_k\|, && \text{Assumption~\ref{Asmp:GradientCond}}\nonumber \\ 
&\leq \max_{j \in \MC{V}} \{N \|\hat{v}^j_k - \bar{v}_k\| +  \alpha_k N \|e^j_k\|\}, &&  \text{Eq.~\ref{Eq:InfoFusionN:v} or Eq.~\ref{Eq:InfoFusionN:vLB} and Triangle Inequality} \nonumber \\ 
&\leq N \left( \max_{j \in \MC{V}} \|q^j_k\| + \alpha_k \Delta \right) \label{Eq:UnrollGradient2Bound} &&  \text{Eq.~\ref{Eq:FusedDisagreement} and }\|e^j_k\| \leq \Delta
\end{align}

We use above expressions to bound the term $\Lambda$, in Eq.~\ref{Eq:IterateRelation2}. We use $\hat{v}^ j_k = \bar{v}_k + q^j_k$ from Eq.~\ref{Eq:FusedDisagreement} and the gradient relation in Eq.~\ref{Eq:UnrollGradient1} to get,
\begin{align}
\Lambda &= -2 \alpha_k \sum_{j=1}^n  (\nabla {f}_j(v^j_k) - e^j_k)^T (\hat{v}^j_k - y)  = 2 \alpha_k \sum_{j=1}^n \left[ (\nabla {f}_j(\bar{v}_k) - e^j_k + l^j_k)^T (y - \bar{v}_k - q^j_k) \right] \nonumber \\
\Lambda &= 2 \alpha_k \left[T_1 + T_2 + T_3\right], \text{  where, } T_1 = \sum_{j=1}^n \left( \nabla {f}_j(\bar{v}_k) - e^j_k \right)^T (y - \bar{v}_k), \nonumber \\
T_2 &= \sum_{j=1}^n \left( \nabla {f}_j(\bar{v}_k) - e^j_k \right)^T (-q^j_k), \ \text{and } T_3 = \sum_{j=1}^n (l^j_k)^T(y - \bar{v}_k - q^j_k) =  \sum_{j=1}^n (l^j_k)^T(y - \hat{v}^j_k). \nonumber 
\end{align}%

\noindent Individually $T_1$, $T_2$ and $T_3$ can be bound as follows, 
{\small \begin{align}
    T_1 &= \sum_{j=1}^n \left(\nabla {f}_j(\bar{v}_k) - e^j_k \right)^T (y - \bar{v}_k) = \nabla f(\bar{v}_k)^T (y-\bar{v}_k) &&  \text{Eq.~\ref{Eq:NB-Agg}, Eq.~\ref{Eq:LB-Agg}}, \sum_{j=1}^n e^j_k = 0 \text{ and } \sum_{j=1}^n \nabla f_j(\bar{v}_k) = \nabla f(\bar{v}_k)\nonumber \\
    &\leq f(y) - f(\bar{v}_k)  \label{Eq:T1Bound}&&  \text{$f(x)$ is convex}\\    
    T_2 &= \sum_{j=1}^n \left( \nabla {f}_j(\bar{v}_k) - e^j_k \right)^T (-q^j_k) \nonumber \\
    &\leq \sum_{j=1}^n \|\nabla {f}_j(\bar{v}_k) - e^j_k\| \|(-q^j_k)\| && \text{Cauchy-Schwarz Inequality}\nonumber \\
    &\leq (L+\Delta) n\max_{j \in \MC{V}} \|q^j_k\|  && \text{Triangle Inequality and Eq.~\ref{Eq:NB-Bounded}, Eq.~\ref{Eq:LB-Bounded}, Assumption~\ref{Asmp:GradientCond}}\nonumber \\
    &\leq (L+\Delta) n\max_{j \in \MC{V}} \|\delta^j_k\| \label{Eq:T2Bound}&& \text{from Eq.~\ref{Eq:qkinequality}}  \\
    T_3 &= \sum_{j=1}^n (l^j_k)^T (y - \hat{v}^j_k) \nonumber \\
    &\leq \max_{j \in \MC{V}} \|l^j_k\| \sum_{j=1}^n \|\hat{v}^j_k - y\| \nonumber \\
    &\leq  N \left( \max_{j \in \MC{V}} \|q^j_k\| + \alpha_k \Delta \right) \sum_{j=1}^n \|\hat{v}^j_k - y\|  &&  \text{from Eq.~\ref{Eq:UnrollGradient2Bound}} \nonumber \\
    &\leq  N \left( \max_{j \in \MC{V}} \|\delta^j_k\| + \alpha_k \Delta \right) \sum_{j=1}^n \|\hat{v}^j_k - y\| && \text{from Eq.~\ref{Eq:qkinequality}} \nonumber  \\
    &\leq N \left( \max_{j \in \MC{V}} \|\delta^j_k\| + \alpha_k \Delta\right)\left[\sum_{j=1}^n \|x^j_k - y\| \right] &&  \text{from Eq.~\ref{Eq:AsideEq1}} \nonumber \\
    &\text{We further use $2 \|a\| \leq 1 + \|a\|^2$ to bound term $T_3$.} \nonumber \\    
    &\leq \frac{N}{2} \left( \max_{j \in \MC{V}} \|\delta^j_k\| + \alpha_k \Delta\right)\left[\sum_{j=1}^n \left( 1 + \|x^j_k - y\|^2 \right) \right]\label{Eq:T3Bound} &&  2 \|a\| \leq 1 + \|a\|^2
\end{align}}%



\noindent We combine the bounds on $T_1, T_2$ and $T_3$ (Eq.~\ref{Eq:T1Bound}, \ref{Eq:T2Bound} and \ref{Eq:T3Bound}) to get,
{\small \begin{align}
\Lambda &\leq 2 \alpha_k (f(y) - f(\bar{v}_k)) + 2 \alpha_k n(L+\Delta) \max_{j \in \MC{V}} \|\delta^j_k\| + \alpha_k N\left( \max_{j \in \MC{V}} \|\delta^j_k\| + \alpha_k \Delta\right)\left[\sum_{j=1}^n \left( 1 + \|x^j_k - y\|^2 \right) \right] \nonumber \\
    \Lambda &\leq -2 \alpha_k \left( f(\bar{v}_k) - f(y) \right) + 2 \alpha_k n(L+\Delta) \max_{j \in \MC{V}} \|\delta^j_k\|  + \alpha_k N \left( \max_{j \in \MC{V}} \|\delta^j_k\| + \alpha_k \Delta\right)\left[n + \eta_k^2 \right] \label{Eq:BDH3}&& \text{Eq.~\ref{Eq:BDH1}}
\end{align}}%

\noindent Recall from Eq.~\ref{Eq:IterateRelation2},
\begin{align}
	 \eta_{k+1}^2 &\leq \eta_k^2 + \alpha_k^2 n (L + \Delta)^2 \underbrace{ - 2 \alpha_k \sum_{j=1}^n (\nabla {f}_j(v^j_k) - e^j_k)^T (\hat{v}^j_k - y)}_{\Lambda}.
\end{align}
We replace $\Lambda$ with its bound from Eq.~\ref{Eq:BDH3}, and use the fact that $\bar{x}_k = \bar{v}_k$ (Eq.~\ref{Eq:IterAverPreserved}) to we replace, $f(\bar{v}_k)$ with $f(\bar{x}_k)$, 
\begin{align}
	\eta_{k+1}^2 &\leq \eta_k^2 + \alpha_k^2 n (L + \Delta)^2 -2 \alpha_k \left( f(\bar{v}_k) - f(y) \right) + 2 \alpha_k n(L+\Delta) \max_{j \in \MC{V}} \|\delta^j_k\|  + \alpha_k N \left( \max_{j \in \MC{V}} \|\delta^j_k\| + \alpha_k \Delta\right)\left[n + \eta_k^2 \right]\nonumber \\
	 &\leq \left(1 + \alpha_k N \left( \max_{j \in \MC{V}} \|\delta^j_k\| + \alpha_k \Delta\right)\right) \eta_k^2  -2 \alpha_k \left( f(\bar{v}_k) - f(y) \right) + 2 \alpha_k n(L+\frac{N}{2}+\Delta) \max_{j \in \MC{V}} \|\delta^j_k\| \nonumber \\
	& \qquad \qquad  + \alpha_k^2 n N \Delta  + \alpha_k^2 n (L + \Delta)^2 \nonumber \\
    &\leq \left(1 + F_k \right)\eta_k^2  - 2 \alpha_k \left(f(\bar{x}_k) - f(y) \right) + H_k, \label{Eq:BDH5} %
\end{align}
where, $F_k = \alpha_k N \left( \max_{j \in \MC{V}} \|\delta^j_k\| + \alpha_k \Delta\right)$ and $H_k = 2 \alpha_k n (L+\frac{N}{2}+\Delta) \max_{j \in \MC{V}} \|\delta^j_k\|  +\alpha_k^2 n \left[N\Delta +(L +\Delta)^2 \right].$
\end{proof}

We first state a claim about asymptotic behavior of iterates $x^j_k$ and correspondingly of $\max_j \|\delta^j_k\|$. The claim is proved in Appendix~\ref{Sec:Appendix-ClaimProof} after the Proof of Theorem~\ref{Th:ConvPPDOP2}.

\begin{claim}[Consensus] \label{Cl:Consensus}
All agents asymptoticaly reach consensus, 
$$\lim_{k \rightarrow \infty} \max_j \|\delta^j_k\| = 0 \text{  and  }\lim_{k \rightarrow \infty} \|x^i_k - x^j_k\| = 0, \forall i, j.$$
\end{claim}
\begin{proof}
See Appendix~\ref{Sec:Appendix-ClaimProof}.
\end{proof}

\subsection{Proof of Theorem~\ref{Th:ConvPPDOP2}} \label{Sec:ProofTh1} 

\begin{proof} 
We prove convergence using Lemma~\ref{Lem:RobSiegConv}. We begin by using the relation between iterates given in Lemma~\ref{Lem:IterateConvRelation} with $y = x^* \in \mathcal{X}^*$, and for $k \geq 1$,
\begin{align}
    \eta_{k+1}^2 &\leq \left(1 + F_k \right)\eta_k^2 - 2 \alpha_k \left(f(\bar{x}_k) - f(y) \right) + H_k \label{Eq:PROOF0} 
\end{align}

We check if the above inequality satisfies the conditions in Lemma~\ref{Lem:RobSiegConv} viz. $\sum_{k=1}^\infty F_k < \infty$ and $\sum_{k=1}^\infty H_k < \infty$. $F_k$ and $H_k$ are defined in Lemma~\ref{Lem:IterateConvRelation} (Eq.~\ref{Eq:BDH5}). Note that $F_k$ and $H_k$ are non-negative, real sequences.  

We first show that $\sum_{k = 1}^\infty \alpha_k \max_{j \in \MC{V}} \|\delta^j_k\| < \infty$ using the expression for state disagreement from average given in Lemma~\ref{Lem:AvgDisagreement1}.
\begin{align}
    \sum_{k = 1}^\infty &\alpha_{k} \max_{j \in \MC{V}} \|\delta^j_{k}\| = \alpha_1 \max_{j \in \MC{V}} \|\delta^j_{1}\| + \sum_{k=1}^\infty \alpha_{k+1} \max_{j \in \MC{V}} \|\delta^j_{k+1}\| \nonumber \\
    &\leq \underbrace{\alpha_1 \max_{j \in \MC{V}} \|\delta^j_{1}\|}_{U_0} + \underbrace{n \theta \max_{i\in\MC{V}} \|x^i_1\| \sum_{k=1}^\infty \alpha_{k+1}\beta^{k}}_{U_1} + \underbrace{n \theta (L+\Delta) \sum_{k=1}^\infty \alpha_{k+1} \sum_{l=2}^k \beta^{k+1-l} \alpha_{l-1}}_{U_2} + \underbrace{2 (L+\Delta) \sum_{k=1}^\infty \alpha_k \alpha_{k+1} }_{U_3}  \\
    &\qquad \qquad \qquad \left(\text{from Lemma~\ref{Lem:AvgDisagreement1}} \right)\nonumber 
\end{align}%

\noindent The first term $U_0$ is finite since $\max_{j \in \MC{V}} \|\delta^j_{1}\|$ and $\alpha_1$ are both finite. The second term $U_1$ can be shown to be convergent by using the ratio test. We observe that,
{\small \begin{align*}
    \limsup_{k \rightarrow \infty} \frac{\alpha_{k+2} \beta^{k+1}}{\alpha_{k+1} \beta^{k}} &= \limsup_{k \rightarrow \infty} \frac{\alpha_{k+2} \beta}{\alpha_{k+1}} < 1 \Rightarrow\sum_{k=1}^\infty \alpha_{k+1} \beta^k < \infty,
\end{align*}}%
since, $\alpha_{k+1} \leq \alpha_k$ and $\beta < 1$. Now we move on to show that $U_2$ is finite. It follows from $\alpha_{k} \leq \alpha_{l}$ when $l \leq k$, and Lemma~\ref{Lem:Ram}  and $\sum_{k} \alpha_k^2 < \infty$,
\begin{align*}
    \sum_{k=1}^\infty \alpha_{k+1} \sum_{l=2}^k \beta^{k+1-l} \alpha_{l-1} \leq  \sum_{k=1}^\infty \sum_{l=2}^k \beta^{k+1-l} \alpha_{l-1}^2 < \infty. 
\end{align*}%
$U_3$ is finite because $U_3 \leq 2(L+\Delta)\sum_{k=1}^\infty \alpha_k^2 < \infty$. Since we have shown, $U_1 < \infty$, $U_2 < \infty$, and $U_3 < \infty$, we conclude $ \sum_{k=1}^\infty \alpha_k \max_{j \in \MC{V}} \|\delta^j_k\| < \infty $.

Clearly, $\sum_{k=1}^\infty F_k< \infty$ and $\sum_{k=1}^\infty H_k< \infty$, since we proved that $\sum_{k=1}^\infty \alpha_k \max_{j \in \MC{V}} \|\delta^j_k\| < \infty$ and we know that $\sum_{k} \alpha_k^2 < \infty$. We can now apply Lemma~\ref{Lem:RobSiegConv} to Eq.~\ref{Eq:PROOF0} and conclude $\sum_{k=1}^\infty 2 \alpha_k \left( f(\bar{x}) - f(x^*)\right) < \infty$. 

We use $\sum_{k=1}^\infty 2 \alpha_k \left( f(\bar{x}) - f(x^*)\right) < \infty$ to show the convergence of the iterate-average to the optimum. Since we know $\sum_{k=1}^\infty \alpha_k = \infty$, it follows directly that $\lim \inf_{k \rightarrow \infty} f(\bar{x}_{k}) = f(x^*) = f^*$ (an alternate proof for this statement is provided later in Appendix~\ref{Sec:Appendix-LIMINFProof}). 


Also note that Lemma~\ref{Lem:RobSiegConv} states that $\eta_k^2$ has a finite limit. Let $\lim_{k \rightarrow \infty} \eta_k^2 = \eta_{x^*}$ ($\forall x^* \in \MC{X}^*$). 
\begin{align*}
\lim_{k \rightarrow \infty} \eta_k^2 &= \lim_{k \rightarrow \infty} \sum_{i=1}^n \|x^i_k - x^*\|^2 = \lim_{k \rightarrow \infty} \sum_{i=1}^n \|\bar{x}_k + \delta^i_k - x^*\|^2 \\
&= \lim_{k \rightarrow \infty} \sum_{i=1}^n \left[ \|\bar{x}_k - x^*\|^2 + \|\delta^i_k\|^2 + 2 (\bar{x}_k - x^*)^T \delta^i_k \right] \\
&= \lim_{k \rightarrow \infty} \left[ n \|\bar{x}_k - x^*\|^2 + \sum_{i=1}^n \|\delta^i_k\|^2 + 2 (\bar{x}_k - x^*)^T \left( \sum_{i=1}^n \delta^i_k \right) \right] \\
&= n \lim_{k \rightarrow \infty} \|\bar{x}_k - x^*\|^2 + \lim_{k \rightarrow \infty} \sum_{i=1}^n \|\delta^i_k\|^2  && \sum_i \delta^i_k = 0 \text{ by definition of $\delta^j_k$} \\
&= n \lim_{k \rightarrow \infty} \|\bar{x}_k - x^*\|^2 \triangleq \eta_{x^*} &&  \lim_{k \rightarrow \infty} \max_{j \in \MC{V}} \|\delta^j_k\| =  0 \text{ from Claim~\ref{Cl:Consensus}}
\end{align*} 

\noindent From the statement above, we know $\lim_{k \rightarrow \infty} \|\bar{x}_k - x^*\| = \sqrt{\frac{\eta_{x^*}}{n}}$. This, along with $\lim \inf_{k \rightarrow \infty} f(\bar{x}_{k}) = f(x^*)$ proves that $\bar{x}_k$ converges to a point in $ \MC{X}^*$.


We know from Claim~\ref{Cl:Consensus} that the agents agree to a parameter vector asymptotically (i.e. $x^j_{k} \rightarrow x^i_{k}, \ \forall i \neq j$ as $k \rightarrow \infty$). Hence, all agents agree to the iterate average. This along with the convergence of iterate-average to an optimal solution gives us that all agents converge to a point in optimal set $\mathcal{X}^*$ (i.e. $x^j_{k} \rightarrow x^* \in \mathcal{X}^*, \ \forall j, \text{ as } k \rightarrow \infty$). This completes the proof of Theorem~\ref{Th:ConvPPDOP2}. 
\end{proof}

\subsection{Alternate Proof of $\liminf_{k \rightarrow \infty} f(\bar{x}_k) = f(x^*) = f^*$} \label{Sec:Appendix-LIMINFProof}
\begin{proof}
We will prove this statement using contradiction. 

\noindent We know that $\sum_k \alpha_k (f(\bar{x}_k) - f(x^*)) < \infty$ and $\sum_k \alpha_k = \infty$. $\{f(\bar{x}_k)\}$ is a sequence of real numbers and we know that liminf always exists for this sequence (its either a real number or symbols $\pm \infty $). Let us assume that $\liminf_{k \rightarrow \infty} f(\bar{x}_k) = f(x^*) + \delta$ for some $\delta > 0$. Note that $\delta$ cannot be less than $0$ since $f(x^*)$ is the minimum. 

\noindent Now we know from the definition of $\liminf$, $\forall \epsilon > 0$, $\exists K_0 \in \mathbb{N}$ such that $\forall k \geq K_0$,
\begin{align*}
\liminf_{l \rightarrow \infty} f(\bar{x}_l) - \epsilon &\leq f(\bar{x}_k), \\
\implies f(x^*) + \delta - \epsilon &\leq f(\bar{x}_k).
\end{align*}
Consider $\epsilon = \delta/2$ and we have, $\exists K_0$ such that, $\forall k \geq K_0$ such that,
\begin{align*}
f(x^*) + \frac{\delta}{2}  \leq f(\bar{x}_k) \implies f(\bar{x}_k) - f(x^*) \geq \frac{\delta}{2}.
\end{align*}

\noindent Let us consider $C \triangleq \sum_{k=1}^\infty \alpha_k (f(\bar{x}_k) - f(x^*)) < \infty$.
\begin{align}
C \triangleq \sum_{k=1}^\infty \alpha_k (f(\bar{x}_k) - f(x^*)) &= \sum_{k=1}^{K_0} \alpha_k (f(\bar{x}_k) - f(x^*)) + \sum_{k=K_0+1}^{\infty} \alpha_k (f(\bar{x}_k) - f(x^*)) \nonumber \\
&\geq \underbrace{\sum_{k=1}^{K_0} \alpha_k (f(\bar{x}_k) - f(x^*))}_{T1} + \underbrace{\sum_{k=K_0+1}^{\infty} \alpha_k \frac{\delta}{2}}_{T2} \label{Eq:ContradictionProofEqNew}
\end{align}
$T1$ is finite since $C$ is finite. $T2$ grows unbounded since $\sum_{k = K_0+1}^\infty \alpha_k = \infty$. Substituting both $T1$ and $T2$ in Eq.~\ref{Eq:ContradictionProofEqNew} we get $C \geq \infty$ in contradiction. Hence $\delta = 0$, implying $\liminf_{k \rightarrow \infty} f(\bar{x}_k) = f(x^*) = f^*$.
\end{proof}

\subsection{Proof of Claim~\ref{Cl:Consensus} (Consensus Claim)} \label{Sec:Appendix-ClaimProof}
\begin{proof} 
We begin with the iterate disagreement relation in Lemma~\ref{Lem:AvgDisagreement1},
{\small \begin{align}
\max_{j \in \MC{V}} \|\delta^j_{k+1}\| &\leq \underbrace{n \theta \beta^{k} \max_{i\in\MC{V}} \|x^i_1\|}_{V_1} + \underbrace{n \theta (L+\Delta) \sum_{l=2}^k \beta^{k+1-l} \alpha_{l-1}}_{V_2} + \underbrace{2 \alpha_k \left( L+ \Delta \right)}_{V_3} \label{Eq:BDH6}
\end{align}%
}
The first term $V_1$ decreases exponentially with $k$. Hence, for any $\epsilon > 0$, $\exists \ K_1 = \lceil{\log_\beta \frac{\epsilon}{3 n \theta \max_{i\in\MC{V}} \|x^i_1\|}}\rceil$ such that, $\forall k > K_1$, we have $V_1 < \epsilon/3$.

For given $\xi = \epsilon (1 - \beta)/6 \beta n \theta (L+\Delta)$, $\exists K_2$ such that, $\alpha_k < \xi$, $\forall \ k \geq K_2$, due to the non-increasing property of $\alpha_k$ and $\sum_k \alpha_k^2 < \infty$. Observe that,
\begin{align*}
\sum_{i = 1}^{k-1}  \left( \alpha_i \beta^{k-i} \right) &=  \underbrace{\left( \alpha_1 \beta^{k-1} + \ \ldots + \alpha_{K_2-1} \beta^{k-K_2+1}\right)}_{A} + \underbrace{ \left(\alpha_{K_2} \beta^{k-K_2} + \ldots + \alpha_{k-1} \beta^1 \right)}_{B}
\end{align*}
\noindent We can bound the terms A and B.
\begin{align}
A &=  \alpha_1 \beta^{k-1} + \alpha_2 \beta^{k-2} + \ldots + \alpha_{K_2-1} \beta^{k-K_2+1} \nonumber \\
&\leq \alpha_1 (\beta^{k-1} +  \ldots + \beta^{k-K_2+1})  &&  \alpha_1 \geq \alpha_i \; \forall \ i \geq 1 \nonumber \\
&\leq \alpha_1 \beta^{k-K_2+1} \left(\frac{1 - \beta^{K_2-1}}{1-\beta} \right) \nonumber \\
&\leq \frac{\alpha_1 \beta^{k-K_2+1}}{1-\beta}  &&  \beta < 1 \label{Eq:ABOUNDS1}\\
B &= \alpha_{K_2} \beta^{k-K_2} + \ldots + \alpha_{k-1} \beta^1 \nonumber \\
& < \xi \beta \left( \frac{1-\beta^{k-K_2}}{1-\beta}\right) \leq \frac{\xi \beta}{1-\beta}  &&  \alpha_i < \xi, \; \forall i \geq K_2  \label{Eq:BBOUNDS1}
\end{align}
The right side of inequality in Eq.~\ref{Eq:ABOUNDS1} is monotonically decreasing in $k$ ($\beta < 1$) with limit $0$ as $k \rightarrow \infty$. Hence $\exists K_{3} > K_2$ such that $A < \epsilon/6n \theta (L+\Delta)$, $\forall \; k \geq K_{3}$. 

We know $\alpha_i < \xi = \epsilon (1 - \beta)/6 \beta n \theta (L+\Delta)$ for all $k \geq {K_2}$. Hence, following Eq.~\ref{Eq:BBOUNDS1}, $B \leq \frac{\xi \beta}{1-\beta} = \frac{\epsilon \beta (1 - \beta)}{6 (1-\beta) \beta n \theta (L+\Delta)} <\epsilon/6n \theta (L+\Delta)$, for all $k \geq K_{2}$. Hence, $\exists K_4 = \max\{K_2,K_3\}$ such that $V_2 = n \theta (L+\Delta) (A+B) < \epsilon/3$ for all $k > K_4$. 


The third term of Eq.~\ref{Eq:BDH6}, $V_3$, decreases at the same rate as $\alpha_k$. Hence, for any $\epsilon>0$, $\exists \ K_6 = \min \{k | \alpha_k < \frac{\epsilon}{6 (L+\Delta)}\}$, such that $\forall k > K_6$, we have $V_3 < \epsilon/3$.

We have convergence based on $\epsilon-\delta$ definition of limits. For any $\epsilon>0$, there exists $K_{\max} = \max\{K_1, K_5, K_6\}$ such that $\max_{j \in \MC{V}} \|\delta^j_k\| \leq V_1 + V_2 + V_3 < \epsilon$ for all $k \geq K_{\max}$. This implies that, $$\lim_{k \rightarrow \infty} \max_{j \in \MC{V}}\|\delta^j_k\| = \lim_{k \rightarrow \infty} \max_{j \in \MC{V}}\|x^j_k - \bar{x}_k\| \leq 0.$$

\noindent Since, $\max_{j \in \MC{V}}\|\delta^j_k\| \geq 0$, the above statement implies $\lim_{k \rightarrow \infty} \max_{j \in \MC{V}}\|\delta^j_k\| = \lim_{k \rightarrow \infty} \max_{j \in \MC{V}}\|x^j_k - \bar{x}_k\| = 0$.

Now note that $\lim_{k \rightarrow \infty} \max_{j \in \MC{V}}\|x^j_k - \bar{x}_k\| = 0 \implies \lim_{k \rightarrow \infty} \|x^j_k - \bar{x}_k\| = 0 \; \forall j$. Hence, we can also show that the following relationship holds,
\begin{align}
\lim_{k \rightarrow \infty} \|x^j_k - {x}^i_k\| &= \lim_{k \rightarrow \infty} \|(x^j_k - \bar{x}_k) + (\bar{x}_k - {x}^i_k)\| \nonumber \\
&\leq \lim_{k \rightarrow \infty} (\|x^j_k - \bar{x}_k\| + \|\bar{x}_k - {x}^i_k\|) && \text{Triangle Inequality} \nonumber \\
&= \lim_{k \rightarrow \infty} \|x^j_k - \bar{x}_k\| + \lim_{k \rightarrow \infty}\|\bar{x}_k - {x}^i_k\| \nonumber \\
&= 0 \nonumber 
\end{align}

Since, $\|x^j_k - x^i_k\| \geq 0$, for all $i,j$, the above statement implies $\lim_{k \rightarrow \infty} \|x^j_k - {x}^i_k\| = 0$.
\end{proof}

\subsection{Proof of Privacy Result (Theorem~\ref{Th:TPriv-2})} \label{Sec:Appendix-T2Proof}

\begin{proof} 

\noindent \underline{\textbf{Proof for P2.}}

\noindent \textbf{(Necessity)}
Let us assume that $\kappa(\MC{G})\leq f$. Hence, if $f$ nodes are deleted (along with their edges) the graph becomes disconnected to form components $I_1$ and $I_2$, (see Figure~\ref{Fig:NecProof1}). Now consider that $f$ nodes, that we just deleted, are compromised by an adversary. 
 Agents generate correlated noise function using Eq.~\ref{Eq:SMCFunctionGen}. All the perturbation functions $s^{j,i}(x)$ shared by nodes in $I_1$ (with agents outside $I_1$) are observed by the coalition. Hence, the true objective function is easily estimated by using,
 \begin{small}
 \begin{align*}
 \sum_{l \in I_1} f_{l}(x) = \sum_{l \in I_1} \widehat{f}_{l} - \left(\sum_{j = 1, i\in I_1}^{j=f} s^{j,i}(x) - \sum_{j = 2, i\in I_1}^{j=f} s^{i,j}(x) \right).
 \end{align*}%
 \end{small}%
 This gives us contradiction. Hence, $\kappa(\MC{G})>f$ is necessary. 
 \begin{figure}[!h]
 \begin{center}
 \centerline{\includegraphics[width=0.4\columnwidth]{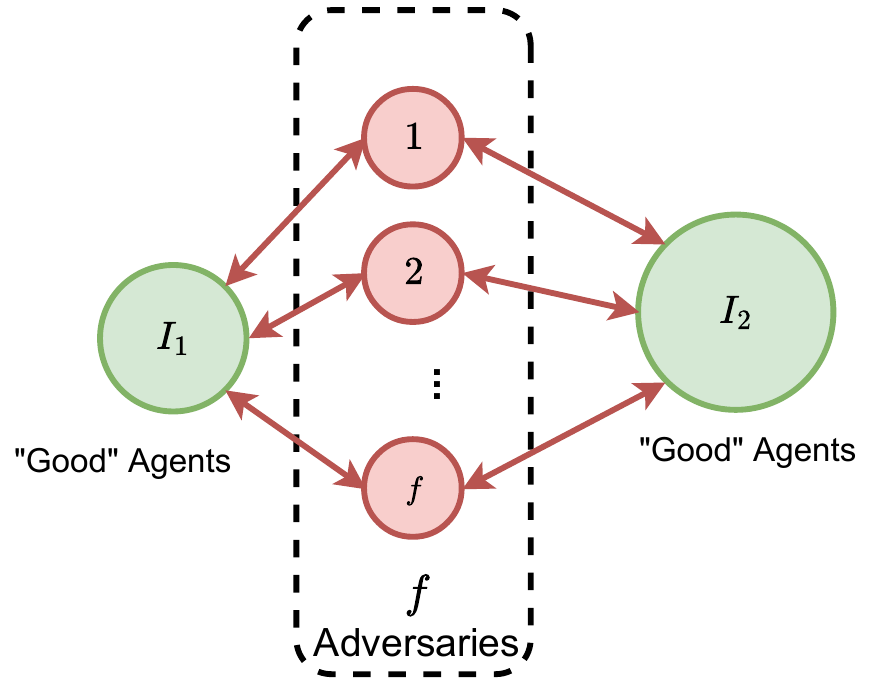}}
 \caption{(Necessity) An example of topology with $\kappa(\MC{G}) = f$.}
 \label{Fig:NecProof1}
 \end{center}
 \end{figure}%

\noindent  \textbf{(Sufficiency)}
We present a constructive method to show that given an execution (and corresponding observations), any estimate of objective functions made by the adversary is equally likely.

We conservatively assume that the adversary can observe the obfuscated functions, $\hat{f}_i(x)$, the private objective functions of corrupted nodes $f_a(x)$ ($a \in \mathcal{A}$) and arbitrary functions transmitted from and received by each of the coalition members, $s^{a,J}$ and $a^{K,a}$ ($J \in \mathcal{N}_{a}$ and $K \text{ such that } a \in \mathcal{N}_{K}$, for all $a \in \mathcal{A}$). Since the corrupted nodes also follow the same protocol (Algorithm~\ref{Algo:PPDOP1}), the adversary is also aware of the fact that the private objective functions have been obfuscated by function sharing approach (Eq.~\ref{Eq:F_New}).

\begin{equation*}
\hat{f}_i(x) = f_i(x) + \sum_{k:i \in \mathcal{N}_k} s^{k,i}(x) - \sum_{j \in \mathcal{N}_i} s^{i,j}(x)
\end{equation*}

Clearly, one can rewrite this transformation approach, using signed incidence matrix of bidirectional graph $\mathcal{G}$ \cite{godsil2013algebraic} \cite{diestel2005graph}.
\begin{align}
\mathbf{\hat{f}} = \textbf{f} + \textbf{B} \textbf{R}.
\end{align}
where, $\mathbf{\hat{f}} = \begin{bmatrix}\hat{f}_1(x), \hat{f}_1(x), \ldots, \hat{f}_S(x)\end{bmatrix}^T$ is a $S \times 1$ vector of obfuscated functions $\hat{f}_i(x)$ for $i = \{1, 2, \ldots, S\}$, and $\textbf{f} = \begin{bmatrix} {f}_1(x), {f}_1(x), \ldots, {f}_S(x)\end{bmatrix}^T$ is a $S \times 1$ vector of private (true) objective functions, $f_i(x)$. $\textbf{B} = \begin{bmatrix}
B_C, -B_C
\end{bmatrix}$, where $B_C$ (of dimension $S \times |\mathcal{E}|/2$) is the incidence matrix of a directed graph obtained by considering only one of the directions of every bidirectional edge in graph $\mathcal{G}$\footnote{This represents an orientation of graph $\mathcal{G}$ \cite{godsil2013algebraic}.}. Each column of $\mathbf{B}$ represents a directed communication link between any two agents. Hence, any bidirectional edge between agents $i$ and $j$ is represented as two directed links, $i$ to $j$, $(i,j) \in \mathcal{E}$ and $j$ to $i$ $(j,i) \in \mathcal{E}$ and corresponds to two columns in $\mathbf{B}$. $\textbf{S}$ represents a $|\mathcal{E}| \times 1$ vector consisting of functions $s^{i,j}(x)$. Each entry in vector $\textbf{S}$, function $s^{i,j}(x)$ corresponds to a column of $\mathbf{B}$ which, in turn corresponds to link $(i,j) \in \mathcal{E}$; and similarly, function $s^{j,i}(x)$ corresponds to a different column of $\mathbf{B}$ which, in turn corresponds to link $(j,i) \in \mathcal{E}$. Note that $\ell^\text{th}$ row of column vector $\textbf{S}$ corresponds to $\ell^\text{th}$ column of incidence matrix $\mathbf{B}$. 

We will show that, two different sets of true objective functions ($\mathbf{f}$ and $\mathbf{f}^o$) and correspondingly two different set of arbitrary functions ($\textbf{S}$ and $\mathbf{G}$), can lead to exactly same execution and observations for the adversary\footnote{$\textbf{f}$ and $\textbf{f}^o$ are dissimilar and arbitrarily different.}. We want to show that both these cases can result in same obfuscated objective functions. That is,
\begin{align}
\mathbf{\hat{f}} = 
\textbf{f} + 
\textbf{B} \textbf{S} = 
\textbf{f}^o + 
\textbf{B} \textbf{G}.
\label{Eq:ProofEq1}
\end{align}

We will show that given any set of private objective functions $\mathbf{f}^o$, suitably selecting arbitrary functions $g^{i,j}(x)$ corresponding to links incident at ``good" agents, it is possible to make $\mathbf{f}^o$ indistinguishable from original private objective functions $\mathbf{f}$, solely based on the execution observed by the corrupted nodes. We do so by determining entries of $\textbf{G}$, which are arbitrary functions that are dissimilar from $s^{i,j}(x)$ when $i$ and $j$ are both ``good". The design $\mathbf{G}$ such that the obfuscated objective functions $\mathbf{\hat{f}}$ are the same for both situations. 

Since corrupted nodes observe arbitrary functions corresponding to edges incident to and from them, we set the arbitrary functions corresponding to edges incident on corrupted nodes as $g^{k,a} = s^{k,a}$ and arbitrary functions corresponding to edges incident away the corrupted nodes as $g^{a,j} = s^{a,j}$ (where $k: a \in \mathcal{N}_{k}$ and $j \in \mathcal{N}_{a}$, for all $a \in \mathcal{A}$). Now, we define $\mathbf{\tilde{G}}$ as the vector containing all elements of $\mathbf{G}$ except those corresponding to the edges incident to and from the corrupted nodes\footnote{The only entries of $\mathbf{G}$, that are undecided at this stage are included in $\mathbf{\tilde{G}}$. These are functions $g^{i,j}$ such that $i$, $j$ are both ``good".}. Similarly, we define $\mathbf{\tilde{B}}$ to be the new incidence matrix obtained after deleting all edges that are incident on the corrupted nodes (i.e. deleting columns corresponding to the links incident on corrupted nodes, from the old incidence matrix $\mathbf{B}$). We subtract $g^{a,j}(x)$ and $g^{k,a}(x)$ ($\forall \ a \in \mathcal{A}$) by subtracting them from $\mathbf{[\mathbf{\hat{f}} - \textbf{f}^o]}$ (in Eq.~\ref{Eq:ProofEq1}) to get effective function difference denoted by $\mathbf{[\mathbf{\hat{f}} - \textbf{f}^o]}_{\rm eff}$ as follows,
\begin{small}
\begin{align}
&\mathbf{[\mathbf{\hat{f}} - \textbf{f}^o]} = \mathbf{B} \mathbf{G} = \mathbf{f} - \mathbf{f}^o +\mathbf{B} \textbf{S},   \ldots (\text{From Eq.~\ref{Eq:ProofEq1}})\\
&\mathbf{[\mathbf{\hat{f}} - \textbf{f}^o]}_{\rm eff} \nonumber \\
&= [\mathbf{\hat{f}} - \mathbf{f}^o] - \sum_{a \in \mathcal{A}} \left[ \sum_{k:a \in \mathcal{N}_k} g^{k,a}(x) - \sum_{j\in \mathcal{N}_a} g^{a,j}(x) \right] \\
&= \mathbf{\tilde{B}} \mathbf{\tilde{G}}, \label{Eq:ProofEq2}
\end{align}%
\end{small}
where, if $d$ entries of $\mathbf{G}$ were fixed\footnote{Total number of edges incident to and from corrupted nodes is $d$. We fixed them to be the same as corresponding entries from $\textbf{S}$, since the coalition can observe them.} then $\mathbf{\tilde{G}}$ is a $(|\mathcal{E}|-d) \times 1$ vector and $\mathbf{\tilde{B}}$ is a matrix with dimension $S \times (|\mathcal{E}| - d)$. The columns deleted from $\mathbf{B}$ correspond to the edges that are incident to and from the corrupted nodes. Hence, $\mathbf{\tilde{B}}$ represents the incidence of a graph with these edges deleted.

We know from the $f$-admissibility of the graph, that $\mathbf{\tilde{B}}$ connects all the non-adversarial agents into a connected component\footnote{The adversarial nodes become disconnected due to the deletion of edges incident on corrupted nodes (previous step).}. Since, the remaining edges form a connected component, the edges can be split into two groups. A group with edges that form a spanning tree over the good nodes (agents) and all other edges in the other group (see Remark~\ref{Ex:STEE} and Figure~\ref{Fig:ExSTEE}). Let $\mathbf{\tilde{B}}_{\rm ST}$ represent the incidence matrix\footnote{Its columns correspond to the edges that form spanning tree.} of the spanning tree and $\mathbf{\tilde{G}}_{\rm ST}$ represents the arbitrary functions corresponding to the edges of the spanning tree. $\mathbf{\tilde{B}}_{\rm EE}$ represents the incidence matrix formed by all other edges and $\mathbf{\tilde{G}}_{\rm EE}$ represents the arbitrary functions related to all other edges.
\begin{align}
\mathbf{[\mathbf{\hat{f}} - \textbf{f}^o]}_{\rm eff} &= \begin{bmatrix}
\mathbf{\tilde{B}}_{\rm ST} & \mathbf{\tilde{B}}_{\rm EE}
\end{bmatrix} \begin{bmatrix}
\mathbf{\tilde{G}}_{\rm ST} \\
\mathbf{\tilde{G}}_{\rm EE}
\end{bmatrix} \\
&=\mathbf{\tilde{B}}_{\rm ST} \mathbf{\tilde{G}}_{\rm ST} + \mathbf{\tilde{B}}_{\rm EE} \mathbf{\tilde{G}}_{\rm EE}.
\end{align}

\noindent We now arbitrary assign functions to elements of $\mathbf{\tilde{G}}_{\rm EE}$ and then compute the arbitrary weights for $\mathbf{\tilde{G}}_{\rm ST}$. We know that the columns of $\mathbf{\tilde{B}}_{\rm ST}$ are linearly independent, since $\mathbf{\tilde{B}}_{\rm ST}$ is the incidence matrix of a spanning tree (cf. Lemma~2.5 in \cite{bapat2010graphs}). Hence, the left pseudoinverse\footnote{$A^\dagger$ represents the pseudoinverse of matrix $A$.} of $\mathbf{\tilde{B}}_{\rm ST}$ exists; and $\mathbf{\tilde{B}}_{\rm ST}^ \dagger \mathbf{\tilde{B}}_{\rm ST} = \mathbb{I}$, giving us the solution for $\mathbf{\tilde{G}}_{\rm ST}$ \footnote{An alternate way to look at this would be to see that $\mathbf{\tilde{B}}_{\rm ST}^T \mathbf{\tilde{B}}_{\rm ST}$ represents the edge Laplacian \cite{zelazo2007agreement} of the spanning tree. The edge Laplacian of an acyclic graph is non-singular and this also proves that left-pseudoinverse of $\mathbf{\tilde{B}}_{\rm ST}$ exists.}. And no adversary can estimate any sum of functions associated with strict subset of ``good" agents.
\begin{align}
\mathbf{\tilde{G}}_{\rm ST} = \mathbf{\tilde{B}}_{\rm ST} ^ \dagger \left[ \mathbf{[\mathbf{\hat{f}} - \textbf{f}^o]}_{\rm eff} - \mathbf{\tilde{B}}_{\rm EE} \mathbf{\tilde{G}}_{\rm EE} \right]. \label{Eq:SolutionEq}
\end{align}

Using the construction shown above, for any $\textbf{f}^o$ we can construct $\textbf{G}$ such that the execution as seen by corrupted nodes is exactly the same as the original problem where the objective is $\textbf{f}$ and the arbitrary functions are $\textbf{S}$. A strong PC coalition cannot distinguish between two executions involving $\textbf{f}^o$ and $\textbf{f}$. Hence, no coalition can estimate $f_i(x)$ ($i \cancel{\in} \mathcal{A}$). 

~

\noindent  \underline{\textbf{Proof for P1.}}

\noindent \textbf{(Necessity)}
The proof of necessity here (for P1) follows the proof of necessity for P2. We prove this statement by contradiction. Assume that a node $i$ has degree $f$ and we have $|\MC{A}| = f$ adversaries. Consider that $I_1 = {i}$ in Figure~\ref{Fig:NecProof1}. 

Agents generate correlated noise function using Eq.~\ref{Eq:SMCFunctionGen}. All the perturbation functions $s^{j,i}(x)$ shared by nodes in $I_1$ (with agents outside $I_1$) are observed by the adversary. Hence, the true objective function is easily estimated by using,
 \begin{align*}
 f_{i}(x) = \widehat{f}_{i} - \left(\sum_{j = 1}^{j=f} s^{j,i}(x) - \sum_{j = 1}^{j=f} s^{i,j}(x) \right).
 \end{align*}%
This gives us contradiction. Hence, degree $>f$ is necessary. 

\noindent \textbf{(Sufficiency)} We can use a construction similar to the sufficiency proof for P2. Instead of considering all objective functions $\mathbf{f}$, we consider $f_i(x)$. 

\end{proof}

\begin{figure*}[!t]
\begin{subfigure}{.48\textwidth}
  \centering
  \includegraphics[width=0.75\linewidth]{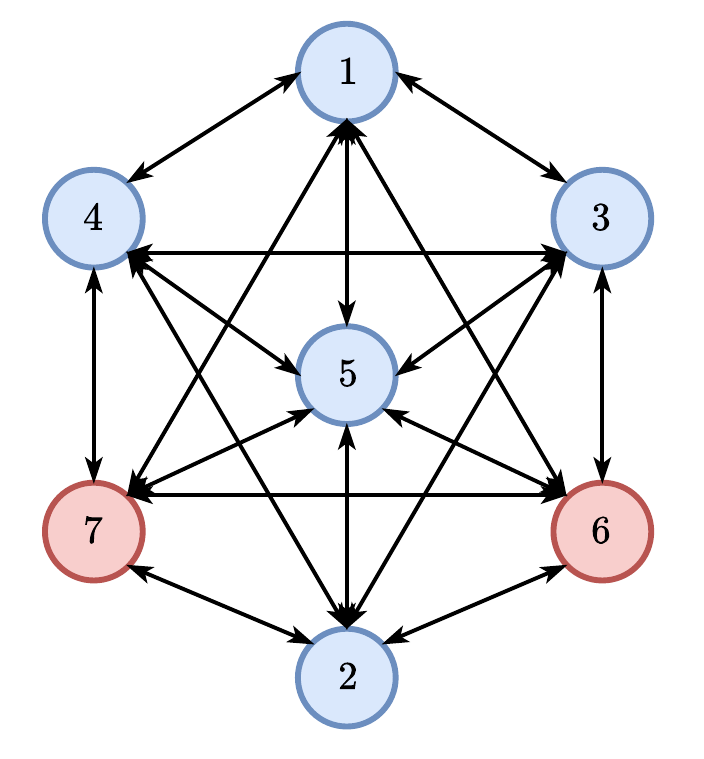}
  \caption{Example }
  \label{Fig:ExSTEE1}
\end{subfigure} \hfill
\begin{subfigure}{.48\textwidth}
  \centering
  \includegraphics[width=0.77\linewidth]{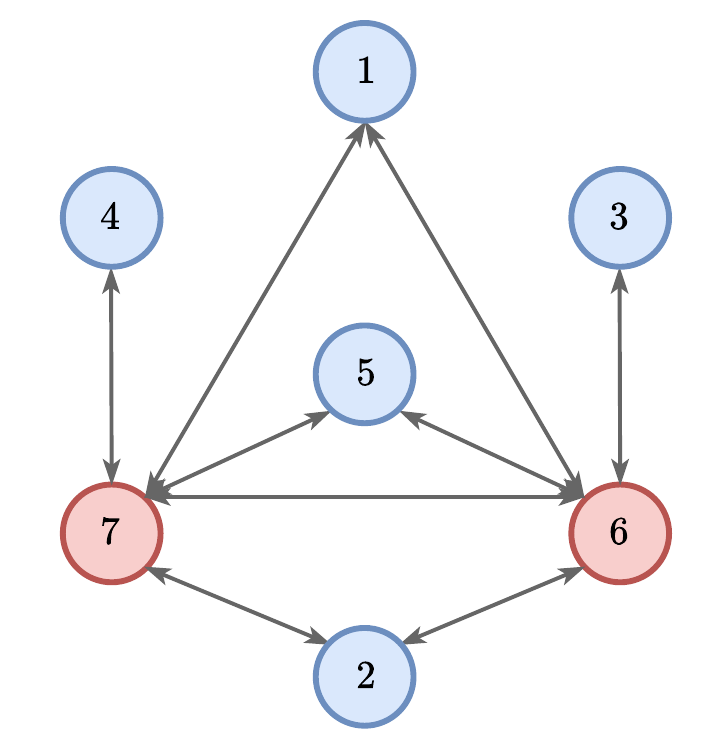}
  \caption{Step-1. Fix $\mathbf{G}_{\mathscr{A},l}$ and $\mathbf{G}_{l,\mathscr{A}}$.}
  \label{Fig:ExSTEE2}
\end{subfigure} \\
\begin{subfigure}{.48\textwidth}
  \centering
  \includegraphics[width=0.75\linewidth]{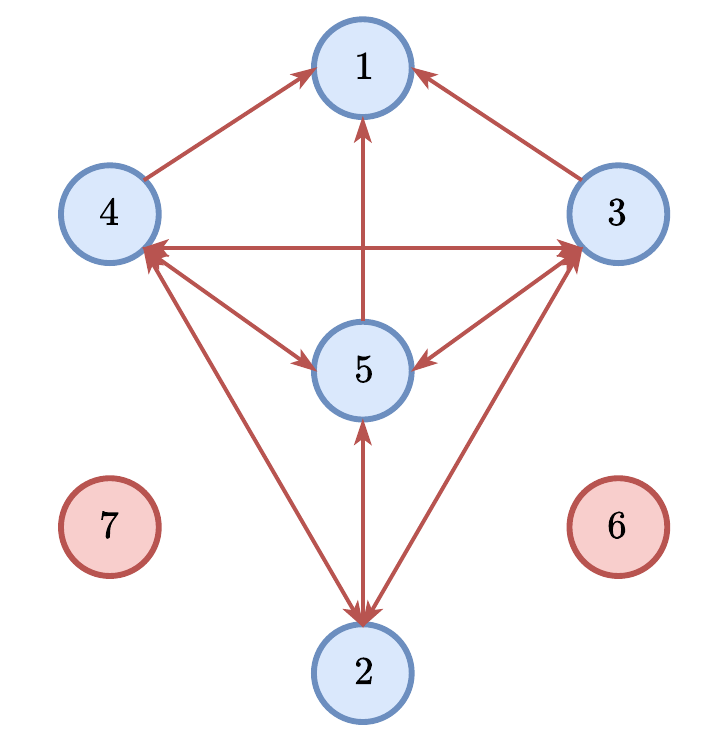}
  \caption{Step-2. Arbitrarily select $\mathbf{G}_{\rm EE}$. }
  \label{Fig:ExSTEE3}
\end{subfigure} \hfill
\begin{subfigure}{.48\textwidth}
  \centering
  \includegraphics[width=0.77\linewidth]{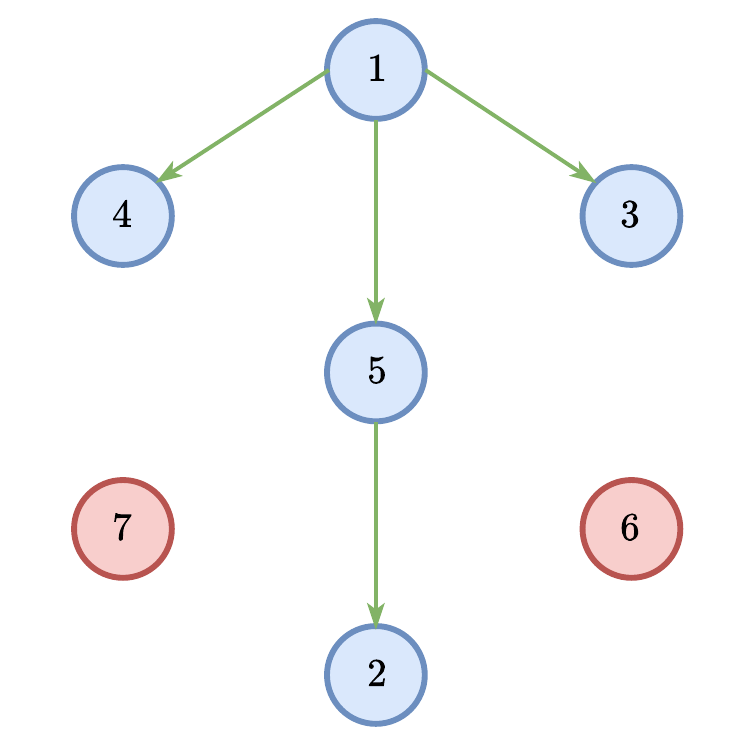}
  \caption{Step-3. Solve for $\mathbf{G}_{\rm ST}$ (Eq.~\ref{Eq:SolutionEq}).}
  \label{Fig:ExSTEE4}
\end{subfigure}
\caption{An example of construction used for proving Theorem~\ref{Th:TPriv-2}. Network has $S=7$ nodes and a adversary with $f=2$ corrupted nodes. The graph topology is $2$-admissible.}
\label{Fig:ExSTEE}
\end{figure*}

\begin{remark} [Method for Constructing $\mathbf{G}$]
\label{Ex:STEE}
We present an example for the construction used in the above proof. Let us consider a system of $S=7$ agents communicating under a topology with $\kappa(\MC{G})>2$ (see Figure~\ref{Fig:ExSTEE1}). An adversary with two corrupted nodes ($\mathcal{A} = \{6, 7\}$, $f=2$) is a part of the system. We can divide the task of constructing $\mathbf{G}$ into three steps - 
\begin{enumerate}
    \item Fix $g^{a,l}$ and $g^{l,a}$ (links incident on corrupted nodes) to be the corresponding entries in $\textbf{S}$,
    \item Arbitrarily select the functions corresponding to non spanning tree edges ($\mathbf{G}_{\rm EE}$), and
    \item Solve for the functions corresponding to the spanning tree ($\mathbf{G}_{\rm ST}$) using Eq.~\ref{Eq:SolutionEq}.
\end{enumerate}

We first, follow the Step 1 and fix $g^{k,a} = s^{k,a}$ and $g^{a,j} = s^{a,j}$ (where $k: a \in \mathcal{N}_{k}$ and $j \in \mathcal{N}_{a}$, for all $a \in \mathcal{A}$). Step 1 follows form the fact that the adversary observes $s^{k,a}$ and $s^{a,j}$, and hence they need to be same in both executions. This is followed by substituting the known entries in $\mathbf{G}$ and subtract them from the left hand side as shown in Eq.~\ref{Eq:ProofEq2}. This corresponds to the deletion of all incoming and outgoing edges from the corrupted nodes. 
The incidence matrix of this new graph is denoted by $\mathbf{\tilde{B}}$. The edges in the new graph can be decomposed into two groups - a set containing edges that form a spanning tree and a set that contains all other edges. This is seen in Figure~\ref{Fig:ExSTEE3} where the red edges are all the remaining links (incidence matrix, $\mathbf{\tilde{B}}_{\rm EE}$); and Figure~\ref{Fig:ExSTEE4} where the green edges form a spanning tree (incidence matrix, $\mathbf{\tilde{B}}_{\rm ST}$) with Agent 1 as the root and all other ``good" agents as its leaves (agents 2, 3, 4, 5). 
\end{remark}

\subsection{Proof of Theorem~\ref{Th:FiniteTimeRes}} \label{Sec:Appendix-FiniteTRes}

\begin{proof}
The proof of Theorem~\ref{Th:FiniteTimeRes} follows from Lemma~\ref{Lem:IterateConvRelation} and a few elementary results on sequences and series.  

We begin our analysis by considering the time weighted average of state $\widehat{x}^j_T = \frac{\sum_{k=1}^T \alpha_k x^j_k}{\sum_{k=1}^T \alpha_k}$, and noting the fact that $\bar{\widehat{x}}_k = \frac{1}{n}\sum_{j=1}^n \widehat{x}^j_k = \frac{\sum_{k=1}^T \alpha_k \bar{x}_k}{\sum_{k=1}^T \alpha_k}$. Also consider the fact that $f(x)$ is convex. We get,
\begin{align}
f(\bar{\widehat{x}}_T) - f^* = f\left(\frac{\sum_{k=1}^T \alpha_k \bar{x}_k}{\sum_{k=1}^T \alpha_k}\right) - f^* \leq \frac{\sum_{k=1}^T \alpha_k f(\bar{x}_k)}{\sum_{k=1}^T \alpha_k} - f^* = \frac{\sum_{k=1}^T \alpha_k \left(f(\bar{x}_k) - f^*\right)}{\sum_{k=1}^T \alpha_k} \label{Eq:FTR1}
\end{align}

\noindent Next, consider $y \in \MC{X}^*$ in Lemma~\ref{Lem:IterateConvRelation}, and we use it to bound the expression in Eq.~\ref{Eq:FTR1}.
\begin{align}
f(\bar{\widehat{x}}_T) - f^* \leq \frac{\sum_{k=1}^T  \left((1+F_k)\eta^2_k - \eta^2_{k+1} + H_k\right)}{2\sum_{k=1}^T \alpha_k} \label{Eq:FTR2}
\end{align}
Canceling the telescoping terms in Eq.~\ref{Eq:FTR2}, we get,
\begin{align}
f(\bar{\widehat{x}}_T) - f^* &\leq \frac{\eta_1^2 - \eta_{T+1}^2 + \sum_{k=1}^T \left(F_k \eta^2_k + H_k\right)}{2\sum_{k=1}^T \alpha_k} \nonumber \\
&\leq \frac{\eta_1^2 + \sum_{k=1}^T \left(F_k \eta^2_k + H_k\right)}{2\sum_{k=1}^T \alpha_k} \label{Eq:FTR3}
\end{align}

If $\alpha_k = 1/\sqrt{k}$, we have from comparison test, $\sum_{k=1}^T \alpha_k \geq \sqrt{T}$. This gives us from Eq.~\ref{Eq:FTR3},
\begin{align}
f(\bar{\widehat{x}}_T) - f^* \leq \frac{\eta_1^2 + \sum_{k=1}^T \left(F_k \eta^2_k + H_k\right)}{2\sqrt{T}} \label{Eq:FTR3}
\end{align}

Let us define the maximum value of $\eta_k^2$ as $D_0$, i.e. $D_0 \triangleq n \max_{x ,y \in \MC{X}} \|x - y\|^2$. Note that due to compactness of $\MC{X} \subseteq \mathbb{R}^D$, the bound $D_0$ is finite. 
\begin{align}
f(\bar{\widehat{x}}_T) - f^* \leq \frac{D_0 + D_0\sum_{k=1}^T F_k + \sum_{k=1}^T H_k}{2\sqrt{T}} \label{Eq:FTR4}
\end{align}

Next we bound $\sum_{k=1}^T F_k$ and $\sum_{k=1}^T H_k$.

\begin{align}
\sum_{k=1}^T F_k &= N \sum_{k=1}^T \alpha_k \max_j \|\delta^j_k\| + N\Delta \sum_{k=1}^T \alpha_k^2 \nonumber\\
&= N \sum_{k=1}^T \alpha_k \max_j \|\delta^j_k\| + N\Delta \sum_{k=1}^T \frac{1}{k} \nonumber\\
&\leq N \sum_{k=1}^T \alpha_k \max_j \|\delta^j_k\| + N \Delta (\log(T) + 1) \qquad \qquad \qquad \left(\sum_{k=1}^T \frac{1}{k} < \log(T)+1\right) \label{Eq:FTR5}\\
\sum_{k=1}^T H_k &=2n(L+N/2+\Delta) \sum_{k=1}^T \alpha_k \max_j \|\delta^j_k\| + n[(L+\Delta)^2 + N \Delta] \sum_{k=1}^T \alpha_k^2 \nonumber \\
&=2n(L+N/2+\Delta) \sum_{k=1}^T \alpha_k \max_j \|\delta^j_k\| + n[(L+\Delta)^2 + N \Delta] \sum_{k=1}^T \frac{1}{k} \nonumber \\
&\leq 2n(L+N/2+\Delta) \sum_{k=1}^T \alpha_k \max_j \|\delta^j_k\| + n[(L+\Delta)^2 + N \Delta] (\log(T)+1) \label{Eq:FTR6}
\end{align}

We use Lemma~\ref{Lem:AvgDisagreement1} to bound $\sum_{k=1}^T \alpha_k \max_j \|\delta^j_k\|$.
\begin{align}
\sum_{k=1}^T &\alpha_k \max_j \|\delta^j_k\| = \alpha_1 \max_j \|\delta^j_k\| + \sum_{k=1}^{T-1} \alpha_{k+1} \max_j \|\delta^j_{k+1}\| \nonumber \\
&\leq \alpha_1 \max_j \|\delta^j_1\| + n\theta \max_i \|x^i_1\| \sum_{k=1}^{T-1} \alpha_{k+1} \beta^{k} + n \theta (L+\Delta) \sum_{k=1}^{T-1} \left( \alpha_{k+1} \sum_{l=2}^{k} \beta^{k+1-l} \alpha_{l-1} \right) \nonumber \\
&\qquad \qquad \qquad  + 2(L+\Delta) \sum_{k=1}^{T-1} \alpha_{k+1} \alpha_k \nonumber\\
&\leq \alpha_1 \max_j \|\delta^j_1\| + n\theta \max_i \|x^i_1\| \sum_{k=1}^{T-1} \alpha_{k+1} \beta^{k} + n \theta (L+\Delta) \sum_{k=1}^{T-1} \left( \alpha_{k+1} \sum_{l=2}^{k} \beta^{k+1-l} \alpha_{l-1} \right) \nonumber \\
&\qquad \qquad \qquad  + 2(L+\Delta) \sum_{k=1}^{T-1}  \alpha^2_k \nonumber\\
&\leq \alpha_1 \max_j \|\delta^j_1\| + n \theta C_0 \max_i \|x^i_1\| + n \theta (L+\Delta) C_1 (\log(T) + 1)+ 2(L+\Delta) (\log(T-1)+1) \nonumber \\
&\leq \alpha_1 \max_j \|\delta^j_1\| + n \theta C_0 \max_i \|x^i_1\| +  2(L+\Delta) + n \theta (L+\Delta) C_1 + \left( n \theta (L+\Delta) C_1 + 2(L+\Delta) \right) \log(T) \label{Eq:FTR7}
\end{align}
where $C_0 = \sum_{k=1}^{T-1} \alpha_{k+1} \beta^k$ and $C_1 = \beta (1-\beta^{T-1})/(1-\beta)$. The existence of $C_0$ and the bound on $C_1$ is presented in the next few lines.

~

\noindent $C_0$ exist and can be proved using ratio test for series convergence. Since, $\alpha_{k+1} \leq \alpha_k$ and $\beta < 1$, we get,
\begin{align*}
    \limsup_{k \rightarrow \infty} \frac{\alpha_{k+2} \beta^{k+1}}{\alpha_{k+1} \beta^{k}} &= \limsup_{k \rightarrow \infty} \frac{\alpha_{k+2} \beta}{\alpha_{k+1}} < 1 \Rightarrow\sum_{k=1}^\infty \alpha_{k+1} \beta^k < \infty.
\end{align*}
Note that $C_0 \leq \sum_{k=1}^{\infty} \alpha_{k+1} \beta^k < \infty$, hence, $C_0$ is a finite constant.

Next, we estimate a bound on, $\sum_{k=1}^{T-1} \left( \alpha_{k+1} \sum_{l=2}^{k} \beta^{k+1-l} \alpha_{l-1} \right)$. 
\begin{align}
\sum_{k=1}^{T-1} \left( \alpha_{k+1} \sum_{l=2}^{k} \beta^{k+1-l} \alpha_{l-1} \right) &\leq \sum_{k=1}^{T-1} \left( \sum_{l=2}^{k} \beta^{k+1-l} \alpha^2_{l-1} \right)  && \alpha_{k+1} \leq \alpha_{l-1}, \forall l \leq k \nonumber \\ 
&= \sum_{k=1}^{T-1} \left( \sum_{l=2}^{k} \frac{\beta^{k+1-l}}{l-1} \right) && \alpha^2_{l-1} = \frac{1}{l-1} \nonumber \\
&= \sum_{k=1}^{T-1} \left(\frac{\sum_{j=1}^{T-1-k} \beta^j}{k}\right) && \text{Rearranging Terms} \nonumber \\
&= \sum_{k=1}^{T-1} \left(\frac{\beta\frac{1-\beta^{T-k}}{1-\beta}}{k}\right) && \text{Sum of geometric series $\sum_{j=1}^{T-1-k} \beta^j = \beta\frac{1-\beta^{T-k}}{1-\beta}$ } \nonumber \\
&\leq \beta \frac{1-\beta^{T-1}}{1-\beta} \sum_{k=1}^{T-1} \frac{1}{k} \nonumber \\
&\leq C_1 \log(T-1) + C_1 && C_1 \triangleq \beta\frac{1-\beta^{T-1}}{1-\beta} \text{ and } \sum_{k=1}^{T-1} \frac{1}{k} \leq \log(T-1)+1 \nonumber \\
&\leq C_1 \log(T) + C_1 && \log(T-1) < \log(T)
\end{align}

\noindent We use the bound on $\sum_{k=1}^T \alpha_k \max_{j \in \MC{V}}\|\delta^j_k\|$ to get the bound on $\sum_{k=1}^T F_k$ (in Eq.~\ref{Eq:FTR5}) and $\sum_{k=1}^T H_k$ (in Eq.~\ref{Eq:FTR6}).
\begin{align}
\sum_{k=1}^T F_k &\leq N \sum_{k=1}^T \alpha_k \max_j \|\delta^j_k\| + N \Delta (\log(T) + 1) \nonumber \\
&\leq C_2 + \left( C_3 + C_4 \right)\log(T) \label{Eq:FTR8n} \\
\end{align}
where,the constants $C_2$, $C_3$ and $C_4$ are defined as,
\begin{align*}
C_2 &= N\Delta + N \left(\alpha_1 \max_j \|\delta^j_1\| + n\theta C_0 \max_i \|x^i_1\| + n \theta (L+\Delta) C_1 + 2(L+\Delta) \right), \\ 
C_3 &= N\Delta, \\
C_4 &= 2N(L+\Delta) + N  n \theta (L+\Delta) C_1. 
\end{align*}
Next we construct a bound on $\sum_{k=1}^T H_k$.
\begin{align}
\sum_{k=1}^T H_k &\leq 2n(L+N/2+\Delta) \sum_{k=1}^T \alpha_k \max_j \|\delta^j_k\| + n[(L+\Delta)^2 + N \Delta] (\log(T)+1) \nonumber \\
&\leq C_5 + (C_6+C_7) \log(T) \label{Eq:FTR9n}
\end{align}
where, the constants $C_5$, $C_6$ and $C_7$ are defined as,
\begin{align*}
C_5 &= 2n(L+N/2+\Delta) \left( \alpha_1 \max_j \|\delta^j_1\| + n\theta C_0 \max_i \|x^i_1\| + n \theta (L+\Delta) C_1 + 2(L+\Delta) \right) + n[(L+\Delta)^2 + N \Delta], \\
C_6 &= n[(L+\Delta)^2 + N \Delta], \\
C_7 &= 2n(L+N/2+\Delta) \left( n \theta (L+\Delta) C_1 + 2(L+\Delta)\right).
\end{align*}

We use the bound from in Eq.~\ref{Eq:FTR8n} and Eq.~\ref{Eq:FTR9n} and combine with finite time relation in Eq.~\ref{Eq:FTR4}.
\begin{align}
f(\bar{\widehat{x}}_T) - f^* \leq \frac{C_8 + (C_9 + C_{10})\log(T)}{2\sqrt{T}}
\end{align}
where, 
\begin{align*}
C_8 &= D_0 C_2+C_5+D_0, \\
C_9 &= D_0C_3+C_6, \\
C_{10} &= D_0 C_4+C_7. 
\end{align*}

\noindent Next we use the Lipschitzness of $f(x)$ to arrive at the statement of Theorem~\ref{Th:FiniteTimeRes}.
\begin{align}
f(\widehat{x}^j_T) - f^* &= f(\widehat{x}^j_T) - f(\bar{\widehat{x}}_T) + f(\bar{\widehat{x}}_T) - f^* \\
&\leq L \|\widehat{x}^j_T - \bar{\widehat{x}}_T\| + \frac{C_8 + (C_9 + C_{10})\log(T)}{2\sqrt{T}} \\
&\leq L \left[ \frac{\sum_{k=1}^T \alpha_k\|x^j_k - \bar{x}_k\|}{\sum_{k=1}^T \alpha_k}\right] + \frac{C_8 + (C_9 + C_{10})\log(T)}{2\sqrt{T}} \\
&\leq L \left[ \frac{\sum_{k=1}^T \alpha_k\max_j \|\delta^j_k\|}{\sum_{k=1}^T \alpha_k}\right] + \frac{C_8 + (C_9 + C_{10})\log(T)}{2\sqrt{T}} \\
&\leq \left[ \frac{C_{11} + C_{12}\log(T)}{\sqrt{T}}\right] + \frac{C_8 + (C_9 + C_{10})\log(T)}{2\sqrt{T}} 
\end{align}
where,
\begin{align*}
C_{11} &= L\alpha_1 \max_j \|\delta^j_1\| + n\theta L C_0 \max_i \|x^i_1\| + 2L(L+\Delta),\\
C_{12} &= n L \theta (L+\Delta) C_1 + 2L(L+\Delta).
\end{align*}

\noindent Observe that $C_9 = \mathcal{O}(\Delta^2)$. We can rewrite the result as,
\begin{align}
f(\widehat{x}^j_T) - f^* = \MC{O}\left( (1+\Delta^2) \frac{\log(T)}{\sqrt{T}}\right) 
\end{align}

\end{proof}

\subsection{Alternate Proof of Eq~\ref{Eq:E6}} 
\label{Sec:Appendix-ConProof}
\begin{proof}
We begin with the consensus update equation in Eq.~\ref{Eq:InfFusALG2:v} (\TT{RSS-NB}) and Eq.~\ref{Eq:InfFusALG2:vLB} (\TT{RSS-LB}), followed by subtracting vector $y \in \mathcal{X}$ on both sides.
\begin{align}
\hat{v}^i_{k} &= \sum_{j=1}^n B_k[i,j] x^j_k   \\
\hat{v}^i_{k} - y &= \sum_{j=1}^n B_k[i,j] (x^j_k - y)  \qquad \ \quad \text{ because } \sum_{j=1}^n B_k[i,j] = 1  
\end{align}
\noindent Let us define $\hat{z}^i_{k} = \hat{v}^i_{k} - y$ and $z^i_{k} = x^i_{k} - y$, $\forall \ i = \{1, 2, \ldots, n\}$. Then we can rewrite the above equation as,
\begin{align}
\hat{z}^i_{k} &= \sum_{j=1}^S B_k[i,j] z^j_k
\end{align}
\noindent We now norm both sides of the equality and use the property that norm of sum is less than or equal to sum of norms.
\begin{align}
\|\hat{z}^i_{k}\| = \|\sum_{j=1}^n B_k[i,j] z^j_k\| \leq \sum_{j=1}^n B_k[i,j] \|z^j_k\| \qquad \qquad  \text{ since }  B_k[i,j] \geq 0
\end{align}
\noindent Squaring both sides in the above inequality followed by algebraic expansion, we get, 
\begin{align}
\|\hat{z}^i_{k}\|^2 &\leq \left(\sum_{j=1}^n B_k[i,j] \|z^j_k\|\right)^2 = \sum_{j=1}^n B_k[i,j]^2 \|z^j_k\|^2 + 2 \sum_{m<j} B_k[i,j] \ B_k[i,m] \ \|z^j_k\| \|z^m_k\| 
\end{align}
\noindent Now we use the property, $a^2 + b^2 \geq 2 a b$ for any $a, b$ in the latter term of the expansion above,
\begin{align}
\|\hat{z}^i_{k}\|^2 &\leq \sum_{j=1}^n B_k[i,j]^2 \|z^j_k\|^2 + \sum_{m < j} B_k[i,j] \ B_k[i,m] \ \left(\|z^j_k\|^2 + \|z^m_k\|^2\right). 
\end{align}
Rearranging and using row stochasticity of $B_k$ matrix we get,
\begin{align}
\|\hat{z}^i_{k}\|^2 &\leq \sum_{j=1}^n \left[ \|z^j_k\|^2 \left(B_k[i,j]^2 + \sum_{m\neq j} B_k[i,j] \ B_k[i,m] \right) \right]\\
&\leq \sum_{j=1}^n \left[ \|z^j_k\|^2 \left(B_k[i,j] \left( B_k[i,j] + \sum_{m\neq j} B_k[i,m] \right) \right) \right] \\
&\leq \sum_{j=1}^n B_k[i,j] \|z^j_k\|^2  \qquad \qquad \text{Row stochasticity of $B_k$ gives } B_k[i,j] + \sum_{m\neq j} B_k[i,m] = 1,\; 
\end{align}
\noindent Summing the inequality over all servers, $i = 1, 2, \ldots, n$, 
\begin{align}
\sum_{i=1}^n \|\hat{z}^i_{k}\|^2 &\leq \sum_{i=1}^n \sum_{j=1}^n B_k[i,j] \|z^j_k\|^2 = \sum_{j=1}^n \left(\|z^j_k\|^2 \left[ \sum_{i=1}^n B_k[i,j] \right] \right) \\
&\leq \sum_{j=1}^n \|z^j_k\|^2  \qquad \qquad \text{Column stochasticity of $B_k$ gives } \sum_{i=1}^n B_k[i,j] = 1,\; 
\end{align}
\noindent This gives us Eq.~\ref{Eq:E6},
\begin{align}
\xi_k^2 = \sum_{j=1}^n \|\hat{v}^j_{k} - y\|^2 \leq \sum_{j=1}^n \|x^j_k - y\|^2 = \eta_k^2  
\end{align}
\end{proof}

\end{document}